\documentclass[a4paper,fleqn]{cas-sc}

\usepackage{soul} 
\usepackage{color, xcolor}
\hypersetup{
	colorlinks=true, 
	linkcolor=black,  
	citecolor=blue,   
	urlcolor=black, 
}
\usepackage[numbers]{natbib}
\usepackage[boxed,commentsnumbered,ruled,linesnumbered]{algorithm2e}

\newtheorem{definition}{Definition}

\newtheorem{strategy}{Strategy}   
\newtheorem{proof}{Proof}    
  
\newtheorem{theorem}{Theorem}

\begin{document}

\shorttitle{SeqRFM: Fast RFM Analysis in Sequence Data} 
\shortauthors{Y. Zheng et al.}


\title [mode = title]{SeqRFM: Fast RFM Analysis in Sequence Data}   

\author[1]{Yanxin Zheng}
\ead{DuoDuOoozyx@gmail.com}
\address[1]{College of Cyber Security, Jinan University, Guangzhou 510632, China}

\author[1]{Wensheng Gan}
\cortext[cor1]{Corresponding author}
\ead{wsgan001@gmail.com}
\cormark[1]

\author[1]{Zefeng Chen}
\ead{czf1027@gmail.com}

\author[1]{Pinlyu Zhou}
\ead{pinlvzhou@gmail.com}

\author[2]{Philippe Fournier-Viger}
\ead{philfv@szu.edu.cn}
\address[2]{College of Computer Science and Software Engineering, Shenzhen University, Shenzhen 518060, China}


\maketitle

\begin{abstract}
    In recent years, data mining technologies have been well applied to many domains, including e-commerce. In customer relationship management (CRM), the RFM analysis model is one of the most effective approaches to increase the profits of major enterprises. However, with the rapid development of e-commerce, the diversity and abundance of e-commerce data pose a challenge to mining efficiency. Moreover, in actual market transactions, the chronological order of transactions reflects customer behavior and preferences. To address these challenges, we develop an effective algorithm called SeqRFM, which combines sequential pattern mining with RFM models. SeqRFM considers each customer's recency (R), frequency (F), and monetary (M) scores to represent the significance of the customer and identifies sequences with high recency, high frequency, and high monetary value. A series of experiments demonstrate the superiority and effectiveness of the SeqRFM algorithm compared to the most advanced RFM algorithms based on sequential pattern mining. The source code and datasets are available at GitHub \url{https://github.com/DSI-Lab1/SeqRFM}. 
\end{abstract}

\section{Introduction}  \label{sec: introduction}

The global proliferation of Internet technologies has fundamentally reshaped information exchange, business conduct, and social engagement on a large scale. The transformation brings forth a dual-edged dynamic for enterprises: a range of new possibilities alongside intricate challenges. As digitalization becomes essential for business activities, the crucial significance of forging and nurturing enduring customer relationships emerges as a pivotal determinant of enterprise prosperity. This realization has led to the development of customer relationship management (CRM) \cite{ngai2009application}. CRM is a strategy and set of practices aiming at effectively managing and nurturing customer relationships to enhance customer satisfaction and loyalty \cite{kotler1974marketing, xu2002adopting}. In CRM, data mining \cite{chen1996data, gan2017data} is a technique that plays a crucial role in extracting valuable information and knowledge from massive and complex data. Data mining involves analyzing large databases using statistics \cite{friedman1998data}, machine learning \cite{bose2001business}, and pattern recognition techniques \cite{jain2000statistical} to discover hidden patterns, associations, and trends. The RFM model \cite{maryani2017clustering} is commonly used in CRM as a data mining technique. To be more specific, recency (R) measures the time since a customer's last interaction or purchase are important, as they reflect customers' level of interest and loyalty. Frequency (F) assesses how often a customer engages with the business, providing insights into their purchasing habits. The monetary (M) represents the amount of money a customer spends, identifying high-value contributors to revenue. The RFM model segments customers into different groups by analyzing and evaluating the three metrics of R, F, and M, such as high-value customers or potential churn customers. High-value customers are those who recently made frequent and substantial purchases. Conversely, potential churn customers are those who have shown little activity or have made low-value transactions over an extended period. The RFM model helps businesses understand the purchasing habits of customers, loyalty, and value, enabling them to understand their customer base and develop targeted marketing strategies. For instance, personalized offers or loyalty programs can be directed toward high-value customers to foster loyalty and encourage repeat purchases. Re-engagement campaigns can be targeted at potential churn customers to regain their interest and business.

RFM analysis methods can be classified into three types. Methods of the first type divide customers by clustering, such as using the $K$-means algorithm \cite{gustriansyah2020clustering}, which is the most commonly used approach. It aims to determine customer value and loyalty by clustering customers based on their characteristics. The second type, known as classification, involves segmenting customers into distinct personas based on the evaluation of the three dimensions of R, F, and M. This method enables targeted marketing strategies by leveraging the unique characteristics of each customer segment. These RFM customers exhibit high loyalty and make significant contributions to the company's long-term profitability. The third type is frequent pattern mining \cite{hu2013knowledge}, which utilizes the scores of R, F, and M as constraints and sets the corresponding minimum thresholds of the three dimensions to identify valuable commodities sequence. The third type of approach represents a valuable application of RFM analysis to marketing strategy optimization. Furthermore, combining data mining with RFM models can provide more comprehensive and accurate data analysis results \cite{chen2012data, dursun2016using}. Accordingly, commodities with higher RFM scores can be considered to have a superior possibility of purchasing and generating profits for the company, indicating a more valuable and meaningful marketing strategy, for it considers all three aspects comprehensively and all possible sequence combinations.

Pattern mining technology stems from frequent itemset mining (FIM). The primary objective of FIM is to identify itemsets \cite{luna2019frequent} that occur with high frequency within large datasets, thereby unveiling potential correlations and useful rules within the data. However, FIM has a significant limitation as it only considers the items purchased by customers without taking into account the sequential ordering of purchases within the same transaction, while the order of purchases is crucial in understanding customer behavior. For instance, a customer is more likely to purchase tires after buying a car, but it is rare for a customer to buy tires before purchasing a car. To solve this problem, sequential pattern mining (SPM) was proposed \cite{fournier2017survey}, which takes sequential ordering information into account. Moreover, this example indicates that the transactions on some specific commodities are derived from other commodities. In real-world market transactions, which unfold chronologically, FIM overlooks the temporal aspects and monetary information encapsulated within the RFM model \cite{chen2009discovering}. For instance, after purchasing a car, the owner may do several tire replacements. Thus, tire purchases may have higher scores of recency and frequency than car sales, while the revenue generated from tire sales is significantly less than that derived from car sales. If we only consider the frequency dimension, high-profit commodities with a low frequency can be easily neglected. Therefore, it is necessary to consider the impact of different dimensions, especially the monetary dimension, in business. For this reason, incorporating SPM with multiple dimensions, that is, along with RFM analysis, holds great practical significance. Some early works \cite{liu2005integrating, chen2006constraint, chen2009discovering} combined the ideas of SPM and RFM, but they did not account for the weights of customer sequences. Hu \textit{et al.} \cite{hu2013knowledge} proposed the RFM-PostfixSpan algorithm, where the RFM values for each item, itemset, and customer sequence are evaluated during the SPM process. Essentially, RFM-PostfixSpan combines the principles of high-utility sequential pattern mining (HUSPM) \cite{le2018pure,gan2021explainable} and utilizes the concept of upper bounds on utility to reduce the search space. In addition, compactness constraints are considered to prevent excessively long time intervals between customer purchases, which match better with realistic consumption patterns. However, there is still significant room for improvement in terms of the efficiency of this algorithm. In addition, the algorithm has some deficiencies in its definitions.

In RFM-PostfixSpan \cite{hu2013knowledge}, certain important issues were also not taken into consideration. For instance, frequency support was defined as the total number of occurrences of patterns, which led to the problem of the Apriori property (non-monotonicity property) not being satisfied. For example, in Fig. \ref{fig: intro_example}, according to the original definition, the frequency support of the pattern $\langle$\{$a$, $d$\}$\rangle$ is 2, while that of pattern $\langle$\{$a$, $d$\}, \{$f$\}$\rangle$ is 3. That is, the frequency support of a pattern may be less than that of a super-pattern. However, this issue is overlooked, which will lead to some errors when pruning the search space. In addition, the monetary support of a pattern was defined as the total number of monetary units of each occurrence, leading to duplication or omission of item utility calculations in some cases, which was not constant with the current standard tasks of HUSPM. For example, if the monetary value of all occurrences of pattern $\langle$\{$a$, $d$\},\{$f$\}$\rangle$ are considered, the items `$a$' and `$d$' in the first itemset and the item `$f$' in the third itemset are calculated repeatedly. Therefore, it is necessary to revise the problem definition to make this task more rigorous and aligned with current techniques for pattern mining.

\begin{figure}[ht]
    \centering
    \includegraphics[clip,scale=0.12]{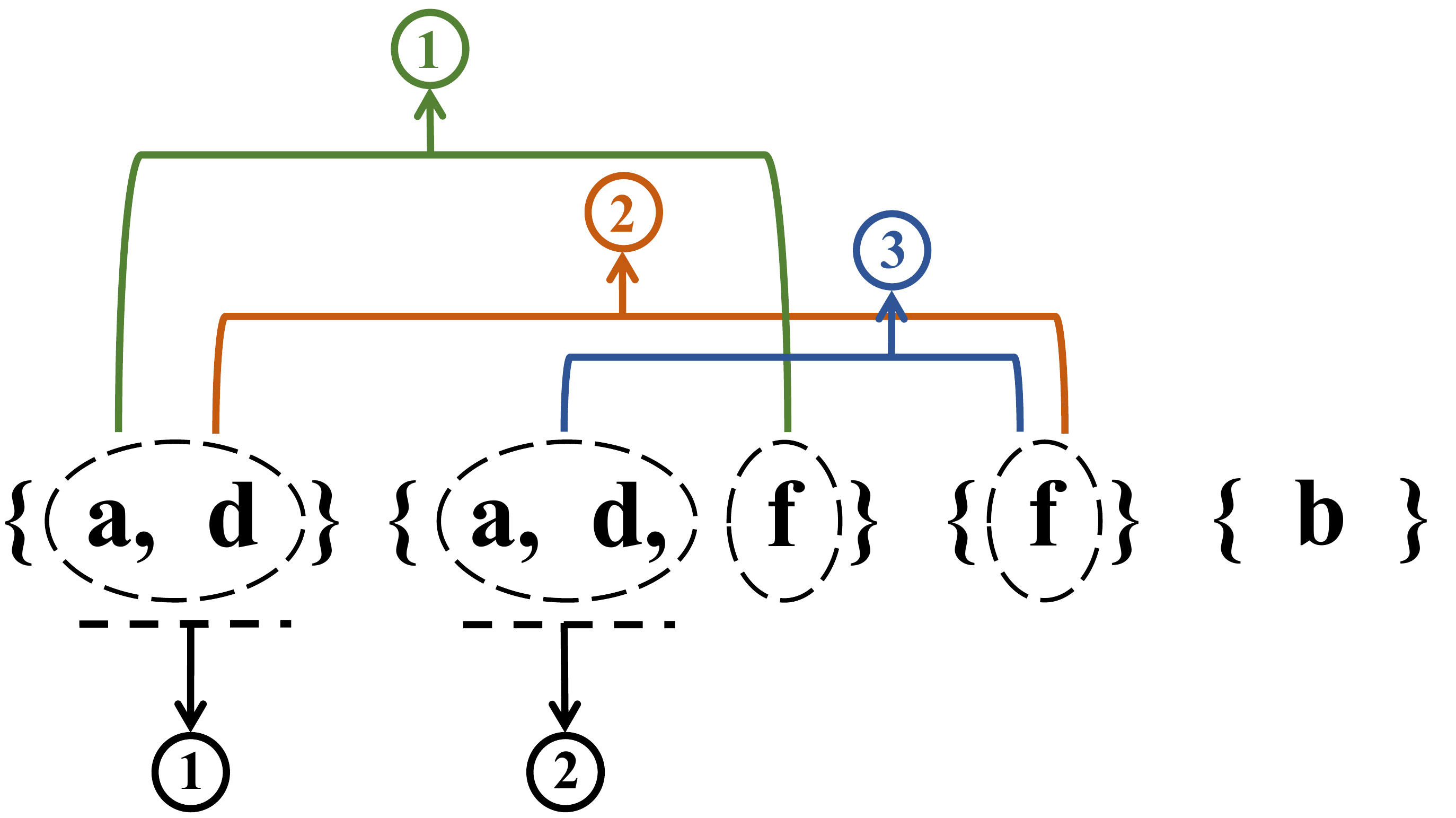}
    \caption{The occurrences of $\langle$\{$a$, $d$\}$\rangle$ and $\langle$\{$a$, $d$\}, \{$f$\}$\rangle$}
    \label{fig: intro_example}
\end{figure}

In response to the aforementioned challenges and to bridge the existing research gap in precision and efficiency of high recency, high frequency, and high monetary value pattern (RFM-pattern) mining in sequence databases, we present SeqRFM, an algorithm inspired by the HUS-Span algorithm \cite{wang2016efficiently}. SeqRFM aims to improve RFM analysis in sequential data. The four key contributions of this paper are as follows:

\begin{itemize}	
    \item  We propose three revised definitions for the three dimensions of RFM-patterns. These definitions are more consistent with existing algorithms for pattern mining.
    
    \item  We develop a novel algorithm named SeqRFM with a novel data structure named RFM-Tree for storing auxiliary information and several pruning strategies for reducing the search space, designed for mining compact RFM-patterns in sequence databases. 

    \item We apply a maximal checking strategy in SeqRFM for mining the maximal RFM-patterns, to obtain a compressed set of mining results.
	
    \item The results of a series of experiments on real and synthetic datasets show that the proposed SeqRFM algorithm provides superior precision and effectiveness compared to state-of-the-art algorithms.
\end{itemize}

Note that this is an extended version of the conference paper \cite{zheng2023fast}. The remaining parts of this paper are organized as follows. Related work is reviewed and summarized in Section \ref{sec: relatedwork}. A series of preliminaries and basic knowledge about RFM-pattern mining are described in Section \ref{sec: preliminaries}. The efficient SeqRFM algorithm and its pruning strategies are detailed in Section \ref{sec: algorithm}. Furthermore, we conducted some experiments, and the experimental results are presented in Section \ref{sec: experiments}. Finally, we conclude and discuss opportunities for future work in Section \ref{sec: conclusion}.

\section{Related Work} \label{sec: relatedwork}

In this section, we review and discuss advances in sequential data analytics, utility mining in sequence data, and the RFM models.

\subsection{Sequential data analytics} 

Sequential pattern mining (SPM) \cite{fournier2017survey,gan2019survey,fournier2022pattern} is an important task in the field of data mining that aims to discover frequently occurring sub-sequences according to the minimum support threshold set by users. Sequential data is usually represented as a series of events in chronological order, such as shopping records, user behavior sequences, and biological sequences. SPM algorithms can be classified into two categories according to the search method: breadth-first search algorithms and depth-first search algorithms. The first SPM algorithm is AprioriAll proposed by Agrawal \textit{et al.} \cite{agrawal1995mining}. This algorithm is based on the famous frequent itemset mining algorithm Apriori \cite{agrawal1994fast}, which efficiently discovers frequent itemsets using a level-wise search strategy, generating candidate itemsets, and pruning those that do not satisfy the minimum support requirement. Subsequently, Agrawal \textit{et al.} \cite{srikant1996mining} improved the AprioriAll algorithm and proposed the GSP algorithm based on a breadth-first search. However, the GSP algorithm requires multiple scans of the database, which can be very time-consuming for large databases. To address this, Zaki \textit{et al.} \cite{zaki2001spade} proposed the SPADE algorithm for depth-first search using the vertical data format, avoiding unnecessary re-scanning and calculation to further improve the efficiency of the algorithm. Nevertheless, such algorithms may generate candidate patterns that do not exist in the database. In response to this issue, pattern growth algorithms have been proposed. For example, Pei \textit{et al.} \cite{pei2004mining} introduced the classical PrefixSpan algorithm, which employs the core idea of generating frequent sequences through prefix growth and reducing the number of candidate patterns using database projection techniques. By applying the prefix growth strategy and data projection technique, the PrefixSpan algorithm can avoid mining frequent sequences in large-scale sequence databases and instead focuses on smaller datasets, significantly reducing the computational complexity. In addition, to satisfy multiple demands, different extensions and constraints based on SPM are proposed, such as negative SPM \cite{dong2018rnsp}, top-$k$ SPM \cite{fournier2013tks}, SPM for mining distinguishing temporal event patterns \cite{duan2017mining}, and SPM with targeted queries \cite{huang2022taspm}.

\subsection{Utility mining in sequence data}

There are many studies on data analysis in sequence data, including frequency-based mining and utility-driven mining \cite{gan2021survey}. It is common in SPM to calculate the support (frequency of occurrences) of sub-sequences to discover frequent sequential patterns. Those are sequences that occur frequently, such as items that are often purchased together in shopping data. However, the support calculation method ignores the actual value and importance of each item, they may not reflect the practical value of each item. While frequent sequential patterns are useful for understanding the co-occurrence relationships of sequences. In HUSPM, each item in each sequence has a corresponding utility value, which indicates the actual value of the pattern. It is common that the higher the utility value, the higher the importance of the pattern. However, the utility of a pattern does not satisfy the Apriori property (anti-monotonic property), that is, the utility of a pattern may be higher than that of its sub-pattern. Therefore, the task of HUSPM has an enormous search space without effective pruning strategies. To cope with this challenge, upper bounds on utility and corresponding pruning strategies were proposed to reduce the search space. To address the above problem, high-utility sequential pattern mining (HUSPM) was introduced by Ahmed \textit{et al.} \cite{ahmed2010novel}, where the sequence-weighted utilization (SWU) is used as an upper bound on the utility to eliminate unpromising sequences and all their super-sequences. Shie \textit{et al.} \cite{shie2011mining} proposed the UMSP algorithm to mine high-utility mobile sequences. However, these algorithms generate a large number of candidate patterns. Wang \textit{et al.} \cite{wang2016efficiently} designed two tight upper bounds to reduce the number of generated candidate patterns. Furthermore, Lin \textit{et al.} \cite{lin2011effective} proposed the Utility Pattern Growth algorithm, which stores information about the utility in a high-efficiency data structure called the Utility Pattern Tree. Gan \textit{et al.} designed a novel data structure called utility-linked-list in the HUSP-ULL algorithm \cite{gan2021fast} and proposed a projection-based utility mining algorithm called ProUM \cite{gan2020proum}.

\subsection{RFM analysis in pattern mining}

RFM analysis consists mainly of three types of methods, but there are fewer studies on pattern mining for RFM analysis than for the other two types. Early work integrated pattern mining with RFM analysis to discover interesting patterns and customer behaviors. Liu \textit{et al.} \cite{liu2005integrating} integrated group decision-making with pattern mining, employing the Analytic Hierarchy Process (AHP), clustering methods, and association rule mining. This work aimed to offer product recommendations tailored to each customer segment. Chen \textit{et al.} \cite{chen2006constraint} introduced the dimensions of compactness and recency as two additional dimensions, besides frequency, into sequential pattern mining (SPM). They proposed a method to discover CFR (compactness, frequency, and recency) patterns. In addition to the dimensions of frequency (F) and recency (R), Chen \textit{et al.} \cite{chen2009discovering} introduced the monetary (M) dimension into SPM, which is the first work for SPM applied in RFM analysis. However, their work did not consider the differences in importance among individual items. To cope with these challenges, the RFM-PostfixSpan algorithm \cite{hu2013knowledge} was proposed to combine an advanced SPM algorithm (PrefixSpan \cite{pei2004mining}) and RFM analysis, and distinguish between the various items. Recently, Wan \textit{et al.} proposed RFM-Miner \cite{wan2022fast1} to discover RFM-pattern in transaction data and RFMUL \cite{wan2022fast} to mine RFM-patterns where the R, F, and M thresholds combined are no less than the user-specified minimum values. Qi \textit{et al.} \cite{qi2022mining} proposed an effective algorithm to mine valuable patterns by applying the fuzzy method to the RFM model. Unfortunately, there have been no recent algorithm optimizations that combine SPM with RFM analysis. With the advancement of HUSPM, new approaches have been developed for computing the monetary (M) dimension, and efficient utility upper bounds and pruning strategies have been devised to reduce unnecessary searching. With the rapid development of e-commerce, the volume of data has been increasing exponentially, necessitating the development of algorithms to address the efficiency challenges in SPM combined with RFM analysis. Building upon the advanced HUSPM algorithm, we are motivated to devise an efficient solution to tackle the computational complexity and scalability issues arising from the integration of RFM analysis in the context of large-scale datasets.

\section{Preliminaries}  \label{sec: preliminaries}

In this section, we provide the formal definitions of high recency, high frequency, and high monetary value compact sequential pattern mining, and we further illustrate their formal problem statement. A monetary and timestamp database (MT-database) is shown in Table \ref{table: database1}.

\begin{table*}[ht]
    \small
	\caption{A monetary and timestamp sequence database}  
	\label{table: database1}
	\centering
	\begin{tabular}{|c|c|}
        \hline
        \textbf{\textit{SID}} & \textbf{MT-sequence} \\ 
        \hline
        $S_1$ &  $\langle $\{($a$, 2) ($d$, 10)\}:100, \{($a$, 19) ($d$, 20) ($f$, 13)\}:74, \{($f$, 40)\}:45, \{($b$, 15)\}:12$\rangle$\\ \hline
        
        $S_2$ &  $\langle $\{($a$, 15) ($d$, 24)\}:100, \{($g$, 50)\}:95, \{($b$, 17)\}:81$ \rangle$\\ \hline
        
        $S_3$ &  $\langle $\{($a$, 11)\}:96, \{($c$, 21)\}:62, \{($a$, 25)\}:58, \{($c$, 19)\}:43$ \rangle$\\ \hline
        
        $S_4$ &  $\langle $\{($b$, 8) ($c$, 15)\}:70, \{($b$, 25) ($c$, 12)\}:62, \{($f$, 15)\}:58$\rangle$\\ \hline

        $S_5$ &  $\langle $\{($a$, 10)\}:98, \{($a$, 21) ($f$, 30)\}:92$\rangle$\\ \hline

        $S_6$ &  $\langle $\{($c$, 15)\}:100, \{($a$, 20) ($d$, 22)\}:92, \{($a$, 20) ($d$, 21)\}:80, \{($a$, 30) ($b$, 10)\}:71$\rangle$\\ \hline      
	\end{tabular}
\end{table*}

\begin{definition}[Item, itemset and sequence]  
  \rm In a database consisting of multiple sequences, a sequence is defined by numerous itemsets, and each itemset contains one or more items. An itemset can be represented as $I$ = \{$i_1$, $i_2$, \dots, $i_n$\}, and the number $n$ of items in the itemset can be represented as $|I|$. That is, $|I|$ = $n$, and the itemset is called an $n$-itemset as well. A sequence $Q$ containing numerous itemsets is denoted as $Q$ = $\langle$$I_1$, $I_2$, \dots, $I_l$$\rangle$. The size of a sequence is $l$ and is denoted as $|Q|$. Moreover, $|Q|$ = $\sum_{k = 1}^{l}$ $|I_k|$, and the sequence $Q$ is called a $l$-sequence.
\end{definition}

For example, consider a sequence $Q$ = $\langle$$I_1$, $I_2$$\rangle$ containing two itemsets, including $I_1$ = \{$a$\} and $I_2$ = \{$a$, $f$\}. The sizes $|I_1|$ and $|I_2|$ of itemsets are 1 and 2. Furthermore, $I_1$ and $I_2$ are 1-itemset and 2-itemset, respectively. Therefore, $Q$ is called a 3-sequence, since $|S_1|$ = 1 + 2 = 3.

\begin{definition}[Monetary and timestamp database]
    \rm  In RFM-pattern mining, an MT-database consists of numerous MT-sequences containing transaction information for customers. Each sequence representing all transactions of a customer in the database is assigned a unique identifier \textit{SID}, which is the index of the sequence in the database. Additionally, each itemset is associated with a timestamp, indicating the time when the set of items was purchased. An MT-item is represented as a tuple in the form of ($i$, $m$), where $i$ $\in$ $I$ and $m$ is the monetary value of the item $i$. A MT-itemset with a size of $n$ is represented by $X$ = $\langle$\{($i_1$,$m_1$) ($i_2$,$m_2$) $\dots$ ($i_n$, $m_n$)\}:$T$, where $T$ indicates the timestamp associated with the itemset. An item can also be represented by a triple of the form  ($i_n$, $m_n$, $T$), called an MT-item, combining a timestamp with the item and a monetary value. In addition, a MT-database can be denoted as $\mathcal{D}$ = $\langle$$Q_1$, $Q_2$, $\dots$, $Q_{\zeta}$$\rangle$, where $\zeta$ is the size of $\mathcal{D}$. 
\end{definition}

For example, as shown in Table \ref{table: database1}, the itemset $\langle$($a$, 15) ($d$, 24)\}:100$\rangle$ is the first itemset in $S_2$. The monetary values of $a$ and $d$ are 15 and 24, and the timestamps of $a$ and $d$ are both 100.

\begin{definition}[sub-sequence and remaining sequence \cite{wang2016efficiently}]
    \rm Given an itemset $X'$ = $\langle$$I'_1$, $I'_2$, $I'_3$, $\dots$, $I'_\xi$ $\rangle$ and an itemset $X$ = $\langle$$I_1$, $I_2$, $I_3$, $\dots$, $I_n$ $\rangle$, where $\xi$ $\leq$ $n$, for 1 $\le$ $\rho_1$ $<$ $\rho_2$ $<$, $\dots$, $\rho_{\xi}$ $\le$ $\xi$, if $I'_1$ = $I_{\rho_1}$, $I'_2$ = $I_{\rho_2}$, $I'_3$ = $I_{\rho_3}$, $\dots$, $I'_{\xi}$ = $I_{\rho_{\xi}}$, $X'$ is defined as a sub-itemset of $X$, denoted as $X'$ $\sqsubseteq$ $X$. Furthermore, given a sequence $Q'$ = $\langle$$X_1$, $X_2$, $X_3$, $\ldots$, $X_\upsilon$ $\rangle$ and a sequence $Q$ = $\langle$$Y_1$, $Y_2$, $Y_3$, $\dots$, $Y_{\Phi}$$\rangle$, where $\upsilon$ $\leq$ $\Phi$, for 1 $\le$ $\rho_1$ $<$ $\rho_2$ $<$, $\dots$, $\rho_{\upsilon}$ $\le$ $\upsilon$, if $X_1$ $\sqsubseteq$ $Y_{\rho_1}$, $X_2$ $\sqsubseteq$ $Y_{\rho_2}$, $\dots$, $X_{\upsilon}$ $\sqsubseteq$ $Y_{\rho_{\upsilon}}$, sequence $Q'$ is considered as a sub-sequence of sequence $Q$, denoted as $Q'$ $\sqsubseteq$ $Q$. $Q$ is a super-sequence of $Q'$, and $Q'$ also can be represented as $Q$($Q'$, $\rho_{\upsilon}$, $Q$), where $\rho_{\upsilon}$ is the ending position of the last itemset in $Q$. Furthermore, the rest of sequence $Q$ from position $\rho_{\upsilon}$ is called the remaining sequence and denoted as $Q/_{(Q', \rho_{\upsilon})}$.
\end{definition}

    \rm Take the database shown in Table \ref{table: database1} as example. Given sequence $Q'$ = $\langle$\{($a$, 2)\}:100, \{($a$, 19)\}:74, \{($f$, 40)\}:45$\rangle$, $Q'$ is a sub-sequence of $S_1$. Its remaining sequence is $S_1/_{(Q', 3)}$ = $\langle$\{($b$, 15)\}:12$\rangle$.

\begin{definition}[Matching]
    \rm Given an itemset $X$ = \{$\eta_1$, $\eta_2$, $\eta_3$, $\dots$, $\eta_n$\} and a MT-itemset $Y$ = \{($i_1$, $m_1$, $t_1$), ($i_2$, $m_2$, $t_2$), ($i_3$, $m_3$, $t_3$), $\dots$, ($i_n$, $m_n$, $t_n$)\}, for each 1 $\le$ $\kappa$ $\le$ $n$, if and only if $\eta_{\kappa}$ = $i_{\kappa}$, $X$ is defined as a match of $Y$, represented by $X$ $\sim$ $Y$. Supposed that there is a sequence $Q'$ = $\langle$$I'_1$, $I'_2$, $I'_3$, $\dots$, $I'_{\xi}$$\rangle$ and a MT-sequence $Q$ = $\langle$$I_1$, $I_2$, $I_3$, $\dots$, $I_{\xi}$$\rangle$, for each 1 $\le$ $\kappa$ $\le$ $\xi$, if and only if $I'_{\kappa}$ $\sim$ $I_{\kappa}$, $Q'$ is defined as a match of $Q$. Similarly, it is denoted as $Q'$ $\sim$ $Q$.
\end{definition}

As shown in Table \ref{table: database1} , given an itemset $X$ = \{$a$, $d$\}  and a MT-itemset $Y$ = \{($a$, 2), ($d$, 10)\}:100, $X$ is defined as a match of $Y$. Similarly, supposing that there is a sequence $Q'$ = $\langle$\{a, d\},  \{a, d, f\}, \{f\}, \{b\}$\rangle$ and a MT-sequence $Q$ = $\langle$\{($a$, 2), ($d$, 10)\}:100, \{($a$, 19), ($d$, 20), ($f$, 13)\}:74, \{($f$,40)\}:45, \{($b$,15)\}:12$\rangle$, $Q'$ is denoted as a match of $Q$.

\begin{definition}[Instance \cite{wang2016efficiently}]
    \rm Given a sequence $Q_1$ = $\langle$$X_1$, $X_2$, $X_3$, $\dots$, $X_\xi$ $\rangle$ and a MT-sequence $Q_2$ = $\langle$$Y_1$, $Y_2$, $Y_3$, $\dots$, $Y_n$$\rangle$, where $\xi$ $\leq$ $n$, for 1 $\le$ $\rho_1$ $<$ $\rho_2$ $<$ $\dots$ $\rho_{\xi}$ $\le$ $n$, if in the sub-sequence $Q'_2$ = $\langle$$Y_{\rho_1}$, $Y_{\rho_2}$, $Y_{\rho_3}$, $\dots$, $Y_{\rho_{\xi}}$$\rangle$, for $\forall$ $Y_k$ $\in$ $Q'_2$ $\wedge$ $Y'$ $\sqsubseteq$ $Y_k$, there always holds $X_k$ $\sim$ $Y'$, then $Q'_2$ is defined as an instance of $Q_1$ in $Q_2$.
\end{definition}

For example, for the database shown in Table \ref{table: database1}, a sequence $Q_1$ = $\langle$\{$a$\}, \{$a$\}, \{$f$\}$\rangle$ and a sequence $Q'_2$ = $\langle$\{($a$, 2)\}:100, \{($a$, 19)\}:74, \{($f$, 40)\}:45$\rangle$, $Q'_2$ is a sub-sequence of $S_1$ and an instance of $Q_1$.

\begin{definition}[Monetary value]
  \rm Given a sequence $Q$ with a unique \textit{SID} = $S$, the total monetary value of $Q$ is defined as $M$($Q$, $S$) = $\sum_{k=1}^{n}$ $m_k$, where $n$ is the size of $Q$ and $m_k$ is the monetary value of the item $i_k$. Assuming that $Q$ contains $m$ instances of a sequence $Q'$, the monetary value of $Q'$ is expressed by $M$($Q'$, $S$) = max\{$M$($Q'_1$, $p_1$, $S$), $M$($Q'_2$, $p_2$, $S$), $\dots$, $M$($Q'_m$, $p_m$, $S$)\}, where $p_m$ is the ending position of the last itemset of the $m$-th instance $Q'$. The total monetary value of $Q'$ in database $\mathcal{D}$ can be represented by $M$($Q'$) = $\sum_{Q_i \in \mathcal{D} \wedge Q' \sqsubseteq Q_i}$ $M$($Q'$,  $S_i$), where $S_i$ is the \textit{SID} of sequence $Q_i$. Assuming that there is a database $\mathcal{D}$ with a size of $n$, the total monetary value of $\mathcal{D}$ is calculated as $M$($\mathcal{D}$) = $\sum_{Q_i \in \mathcal{D}}$ $M$($Q_i$, $S_i$). Given the relative minimum monetary threshold $\gamma$ and the total monetary value $M$($\mathcal{D}$) of the database, if $M$($Q'$) $\ge$ $\gamma$ $\times$ $M$($\mathcal{D}$), $Q'$ is called a M-pattern. Through this redefinition, the monetary value of each pattern avoids redundant recalculation by taking the maximum value, which is more in line with current HUSPM tasks.
\end{definition}
    
For example, as shown in Table \ref{table: database1}, the total monetary value of the first itemset $\langle$\{(a, 2) (d, 10)\}:100$\rangle$ is equal to 2 + 10 = 12. Moreover, the total monetary value of $S_1$ is the sum of all itemsets, which is (2 + 10) + (19 + 20 + 13) + (40 + 15) = 119. The calculations of other sequences are similar to $S_1$. Therefore, the total monetary value of the database $\mathcal{D}$ can be calculated as $\sum_{Q \sqsubseteq \mathcal{D}}$ = 119 + 106 + 76 + 75 + 61 + 138 = 575. The item $\langle$\{$a$\}$\rangle$ occurs three times in sequence $S_6$, with corresponding timestamps of 92, 80, and 71 and respective item monetary values of 20, 20, and 30. Thus, $M$($\langle$\{$a$\}$\rangle$, $S_6$) = \textit{max}\{20, 20, 30\} = 30. Moreover, $\langle$\{$a$\}$\rangle$ also occurs in $S_1$, $S_2$, $S_3$ and $S_5$. $M$($\langle$\{$a$\}$\rangle$) = $M$($\langle$\{$a$\}$\rangle$, $S_1$) + $M$($\langle$\{$a$\}$\rangle$, $S_2$) + $M$($\langle$\{$a$\}$\rangle$, $S_3$) + $M$($\langle$\{$a$\}$\rangle$, $S_5$) + $M$($\langle$\{$a$\}$\rangle$, $S_6$) = 19 + 15 + 25 + 10 + 30 = 99. Supposed that the relative minimum monetary threshold $\gamma$ is 0.15, $M$($\langle$\{$a$\}$\rangle$) = 99 $\ge$ (0.15 $\times 575$), then $\langle$$a$$\rangle$ is a M-pattern. 

\begin{definition}[Recency value]
  \rm  For a sequence $Q$ that includes several instances of a sub-sequence $Q'$, the current timestamp of $Q'$ is denoted as \textit{CT}($Q'$, $Q$), which is the first itemset in $Q$.
  The recent timestamp of an instance $Q'$ is represented as \textit{RT}($Q'$, $Q$), which is the timestamp of the first itemset of $Q'$ in $Q$. Thus, for any instance $Q'$ in $Q$, its current timestamp is the same.
  Given a decay speed $\delta$ between 0 and 1, the recency value of an instance $Q'$ can be represented as $R$($Q'$, $p_i$, $Q$) = (1 - $\delta$)$^{CT - RT_i}$, where $p_i$ is the ending position of the last itemset and $RT_i$ is the timestamp of the $i$-th instance $Q'$. 
  For multiple instances of $Q'$ in $Q$, the recency value of $Q'$ in a sequence $Q$ is denoted as $R$($Q'$, $Q$) = \textit{max}\{$R$($Q'$, $p_1$, $Q$), $R$($Q'$, $p_2$, $Q$), $\dots$, $R$($Q'$, $p_m$, $Q$)\}, where $m$ is the total number of instances in $Q$. 
  Moreover, the recency value of $Q'$ in $\mathcal{D}$ is denoted as $R$($Q'$) = $\sum_{Q \in \mathcal{D} \wedge Q' \sqsubseteq Q}$ $R$($Q'$, $Q$). Furthermore, given an absolute minimum recency value threshold $\alpha$, if $R$($Q'$) $\ge$ $\alpha$, then the sequence $Q'$ is called a R-pattern. To maintain consistency with the definition of the monetary dimension, we also redefine the recency dimension.
\end{definition}

For example, consider $\langle$\{$c$\}, \{$a$\}$\rangle$ in Table \ref{table: database1}. It occurs in $S_3$ and $S_6$. Suppose that the decay speed $\delta$ is 0.01 and the absolute minimum recency threshold $\alpha$ is 1.2. In $S_3$, the current timestamp and recent timestamp are 96 and 62, so $R$($\langle$\{$c$\}, \{$a$\}$\rangle$, $S_3$) = \textit{max}\{(1 - 0.01)$^{96 - 62}$\} $\approx$ 0.71. In $S_6$, $\langle$\{$c$\}, \{$a$\}$\rangle$ occurs three times with their current timestamp being 100, and $R$($\langle$\{$c$\},\{$a$\}$\rangle$, $S_6$) = \textit{max}\{(1 - 0.01)$^{100 - 100}$, (1 - 0.01)$^{100 - 100}$, (1 - 0.01)$^{100 - 100}$\} = 0.99. Thus, $R$($\langle$\{$c$\}, \{$a$\}$\rangle$) = 0.71 + 0.99 = 1.7 $\ge$ 1.2. It is a R-pattern.

\begin{definition}[Frequency value]
  \rm  Given a sequence $Q'$, if an instance of $Q'$ occurs in $Q$, the frequency value of $Q'$ can be represented as $F$($Q'$, $Q$) = 1. Moreover, the total frequency of $Q'$ in $\mathcal{D}$ is denoted as $F$($Q'$) = $\sum_{Q \in \mathcal{D} \wedge Q' \sqsubseteq Q}$ $F$($Q'$, $Q$). Suppose that the relative minimum frequency threshold is $\beta$, and the number of sequences in the database is denoted as $|\mathcal{D}|$, if $F$($Q'$) $\ge$ $\beta$ $\times$ $|\mathcal{D}|$, the sequence $Q'$ is considered a F-pattern. The redefined definition conforms to the sequential data mining task and satisfies the downward closure property.
\end{definition}

For example, let us assume that the relative minimum frequency threshold $\beta$ is 0.5. In Table \ref{table: database1},  the number of sequences is 6, $|\mathcal{D}|$ equals 6, thus the minimum frequency threshold is $\beta$ $\times$ $|\mathcal{D}|$ = 3. As for $\langle$\{$a$\}$\rangle$, it occurs in $S_1$, $S_2$, $S_3$, $S_5$ and $S_6$, $F$($\langle$\{$a$\}$\rangle$) = 5 $\ge$ 3. However, $\langle$\{$g$\}$\rangle$ only occurs in $S_2$, and hence $F$($\langle$\{$g$\}$\rangle$) = 1 $<$ 3. Thus, $\langle$\{$a$\}$\rangle$ is a F-pattern, but $\langle$\{$g$\}$\rangle$ is not a F-pattern.
    
\begin{definition}[Extension and extension item \cite{chen2023uucpm}]
  \rm  Given a $l$-sequence $\kappa$ and an item $i$, if a ($l+1$)-sequence $\epsilon$ is generated by appending an item ${i}$ to the last itemset of $\kappa$, this process is called an I-Extension. We denote it by the formula $\epsilon$ = $\kappa$ $\bigoplus$ $i$. On the other hand, if $\epsilon$ is generated by the process of considering an item $i$ as a new itemset and appending it after the last itemset of $\kappa$, this process is called an S-Extension. Similarly, we denote it by the formula $\epsilon$ = $\kappa$ $\bigotimes$ $i$. In addition,  for the extended pattern $\epsilon$, $\kappa$ is called its prefix pattern.
\end{definition}

For example, as $S_1$ shown in Table \ref{table: database1} and a 2-sequence $\langle$\{($a$, 19) ($d$, 20)\}: 74$\rangle$, the 3-sequence $\langle$\{($a$, 19) ($d$, 20) ($f$, 13)\}: 74$\rangle$ is formed by I-extension, and the 3-sequence $\langle$\{($a$, 19) ($d$, 20)\}: 74, \{($f$, 40)\}: 45$\rangle$ is obtained by an S-extension.

\begin{definition}[Compactness constraint \cite{hu2013knowledge}]
  \rm   Given a maximum time-span length $\theta$, assuming that there is a sequence $Q$ = $\langle$$I_1$, $I_2$, $\dots$, $I_n$$\rangle$, the time-span length of $Q$ can be represented as \textit{TSP}($Q$) = \textit{T}($I_1$) - \textit{T}($I_n$), where $T$($I_1$) and $T$($I_n$) represent the timestamps of the 1-st itemset and the $n$-th itemset. Moreover, if \textit{TSP}($Q$) $\le$ $\theta$, $Q$ is named a compact sequence. 
\end{definition}

For example, the timestamp of $S_6$ in Table \ref{table: database1} is 100. The time-span length of $\langle$\{($c$, 15)\}:100, \{($a$, 20) ($d$, 22)\}:92$\rangle$ is 100 - 92 = 8. The time-span length of $\langle$\{($c$, 15)\}:100, \{($a$, 20) ($d$, 22)\}:92, \{($a$, 20) ($d$, 21)\}:80$\rangle$ is 100 - 80 = 20. Given a maximum span length $\theta$ = 15, there is 100 - 92 = 8 $<$ 15 and 100 - 80 = 20 $\geq$ 15. Thus, $\langle$\{($c$, 15)\}:100, \{($a$, 20) ($d$, 22)\}:92$\rangle$ is a compact sequence, while $\langle$\{($c$, 15)\}:100, \{($a$, 20) ($d$, 22)\}:92, \{($a$, 20) ($d$, 21)\}:80$\rangle$ is not a compact sequence.

\textbf{Problem statement:} Let there be a monetary and timestamp sequence database $\mathcal{D}$, a decay rate $\delta$, a absolute minimum recency threshold $\alpha$, a relative minimum frequency threshold $\beta$, a relative minimum monetary threshold $\gamma$ and a maximum time-span length $\theta$. If a sequential pattern $P$ in $\mathcal{D}$ satisfies four conditions including $R$($P$) $\ge$ $\alpha$, $F$($P$) $\ge$ $\beta$ $\times$ $|\mathcal{D}|$, $M$($P$) $\ge$ $\gamma$ $\times$ $M$($\mathcal{D}$), and \textit{TSP}($P$) $\le$ $\theta$ , then $P$ is said to be a compact RFM-pattern. The task of the algorithm in this paper is to search all compact RFM-patterns (\textit{RFMs}) in the given database $\mathcal{D}$.

To illustrate the task of mining \textit{RFMs}, we present an example using the MT-database shown in Table \ref{table: database1}, where $M$($\mathcal{D})$ is 575. Suppose that $\delta$ = 0.1, $\alpha$ = 1.44, $\beta$ = 0.2, $\gamma$ = 0.25, and $\theta$ = 60, thus, the absolute minimum recency is 1.44, the absolute minimum frequency is 2, and the absolute minimum monetary value is 143.75. Moreover, all the mined \textit{RFMs} are $\langle$\{$a$, $d$\}, \{$g$\}$\rangle$ and $\langle$\{$a$\}, \{$g$\}$\rangle$ with recency values of 2.00 and 3.00, frequency values of 2 and 3, and monetary values of 168 and 164, respectively.

It is worth noting that it is hard to take the above four parameters into account at the same time in the experiments \cite{hu2014discovering}. Due to the importance of frequency and utility in pattern mining, the focus is on the frequency and monetary dimensions in the following. Moreover, the time-span length is considered an important parameter because of the aim of mining compact patterns.

\section{Proposed SeqRFM algorithm}
\label{sec: algorithm}

\subsection{Data structure} \label{sec:datastructure}

Inspired by HUS-Span \cite{wang2016efficiently}, we design a tree structure called RFM-Tree to mine the compact \textit{RFMs}. The RFM-Tree starts with an empty root node ``$\langle \rangle$", and each node $t$ contains two parts, including \textit{t.seq} and \textit{t.mtchain}, where \textit{t.seq} is the sequence and \textit{t.mtchain} is the MT-chain, which stores some auxiliary information to reduce the search space and efficiently calculate the related information of the recency and monetary dimensions.

\begin{definition}[The auxiliary monetary information \cite{wang2016efficiently}]
  \rm  For an I-extension or S-extension that results in an extended pattern $\epsilon$ in a $l$-sequence $Q$, the maximum monetary value of the prefix pattern $\kappa$ is denoted and defined as  \textit{AUXM($\kappa$, $Q$)} = \textit{max}\{$M$($\kappa$ $j_i$, $Q$)$|$$\forall 1 \leq j_1 \leq j_2 \leq \dots \leq j_n \leq n\}$, where $n$ is the number of instances of $\kappa$. 
  The auxiliary monetary information (\textit{AUXM}) helps in computing the monetary value of the prefix pattern $\kappa$ by storing the auxiliary information of the maximum monetary value.
\end{definition}

For example, for sequence $Q$ = $S_6$ shown in Table \ref{table: database1}, there are two instances of $\langle$\{$ad$\}$\rangle$ in $S_6$, including $Q$($\langle$\{$ad$\}$\rangle$, 2, $S_6$) and $Q$($\langle$\{$ad$\}$\rangle$, 3, $S_6$). The \textit{AUXM($\langle$\{$ad$\}$\rangle$, $p$, $S_6$)} = \textit{max}\{20 + 22, 20 + 21\} = \textit{max}\{42, 41\} = 42.

\begin{definition}[The monetary value of a remaining sequence \cite{wang2016efficiently}]
  \rm For an I-extension or S-extension and a prefix pattern $\kappa$ in $Q$, let \textit{RM($\kappa$, $p_i$, $Q$)} = $\sum_{i = p_i + 1}^{l}$ $m_i$, where $p_i$ is the ending position of the last item in $Q$. The residual monetary value (\textit{RM}) is utilized to document the total monetary value of the remaining sequence.
\end{definition}

For example, for $Q$ = $S_6$ shown in Table \ref{table: database1}, there are three instances of $\langle$\{$c$\}\{$a$\}$\rangle$ in $S_6$, encompassing $Q$($\langle$\{$c$\}\{$a$\}$\rangle$, 2, $S_6$), $Q$($\langle$\{$c$\}\{$a$\}$\rangle$, 4, $S_6$), and $Q$($\langle$\{$c$\}\{$a$\}$\rangle$, 6, $S_6$). Their monetary values for the remaining sequence are: \textit{RM($\langle$\{$c$\}\{$a$\}$\rangle$, 2, $S_6$)} = 22 + 20 + 21 + 30 + 10 = 103, \textit{RM($\langle$\{$c$\}\{$a$\}$\rangle$, 4, $S_6$)} = 21 + 30 + 10 = 61, and \textit{RM($\langle$\{$c$\}\{$a$\}$\rangle$, 6, $S_6$)} = 10, respectively.

In previous work, the RFM-PostfixSpan algorithm \cite{hu2013knowledge} performed redundant operations to compute the recency value, frequency value, and monetary value of a pattern that occurs multiple times in its super-sequence. Aiming at improving search efficiency, we design the MT-chain structure to record the needed information, which has five fields: {\textit{SID}: The \textit{SID} of the sequence in the database. \textit{TID}: The extension position in the sequence. \textit{TS}: The recent timestamp of the prefix pattern $\kappa$. \textit{AUXM}: The maximum monetary value of the prefix pattern at the ending position. \textit{RM}: The total monetary value of the remaining sequence at the ending position.

\begin{figure}[ht]
    \centering
    \includegraphics[clip,scale=0.11]{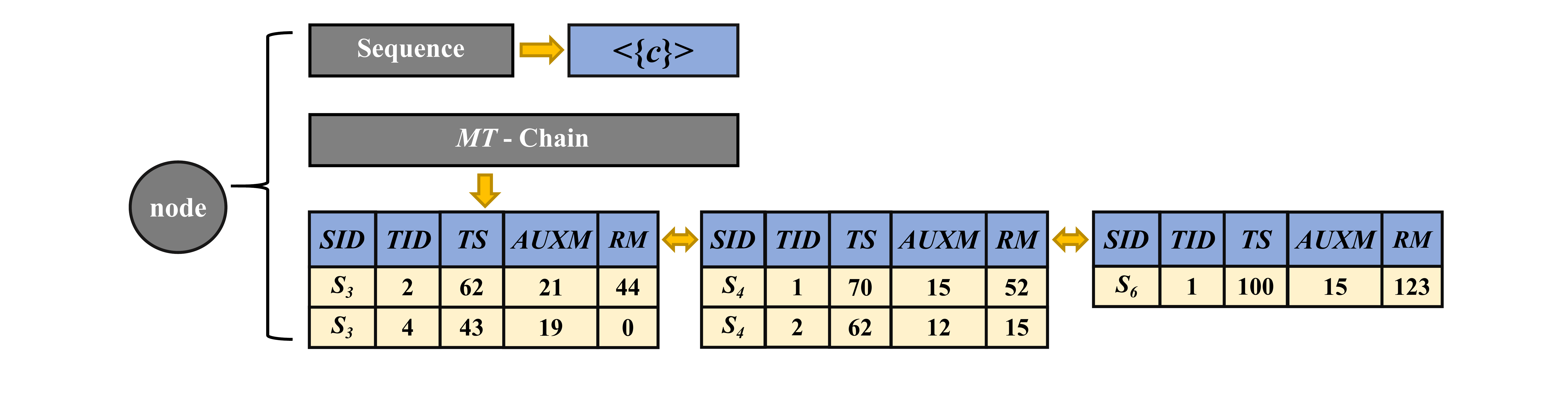}
    \caption{MT-chain for sequence $\langle$\{$c$\}$\rangle$.}
    \label{fig: MT-chain of c}
\end{figure}

According to the MT-chain, the values of three dimensions can be quickly calculated. Assuming that there are multiple instances of \textit{t.seq} in the sequence $Q$, all the elements that should be added to the MT-chain are calculated and added. As for recency dimension, the sequence $Q$ can be acquired from the MT-chain according to the \textit{SID} field, denoted as $Q_s$, and the recent timestamp can be acquired from the \textit{TS} field in the MT-chain element, denoted as \textit{RT}(\textit{t.seq}, $Q_s$). The current timestamp of the instance $Q'$ in sequence $Q$ is represented as \textit{CT}($Q'$, $Q$). Considered each \textit{t.seq}, $R$(\textit{t.seq}) = $\sum_{Q_s \in t.mtchain}$ \textit{max}\{(1-$\delta$)$^{CT-RT_1}$, (1-$\delta$)$^{CT-RT_2}$, $\dots$, (1-$\delta$)$^{CT-RT_m}$\}, where $m$ is the number of instances in sequence $Q_s$. As for monetary dimension, $M$(\textit{t.seq}) = $\sum_{Q_s \in t.mtchain}$ \textit{max}\{\textit{AUXM$_1$}, \textit{AUXM$_2$}, \textit{AUXM$_3$}, $\dots$, \textit{AUXM$_m$}\}. As for the frequency dimension, if a pattern exists the monetary value in the sequence, this indicates the pattern occurs in the sequence. Thus, the number of elements is the $F$(\textit{t.seq}) in the database.

Taking the 1-sequence $\langle$\{$c$\}$\rangle$ as an example, the process of generating MT-chain is shown in Fig. \ref{fig: MT-chain of c}. It can be seen that $\langle$\{$c$\}$\rangle$ occurs twice in $S_3$, twice in $S_4$, and once in $S_6$. In $S_3$, $\langle$\{$c$\}$\rangle$ primarily occurs at the second itemset, so that its \textit{TID} is 2, and timestamp (\textit{TS}) is 62. Besides, \textit{AUXM}($\langle$\{$c$\}$\rangle$, $S_3$) is 21, and \textit{RM} could be calculated by the formula \textit{RM} = 25 + 19 = 44. Then, for its second occurrence at the fourth itemset, its \textit{TID} and \textit{TS} are 4 and 43, and \textit{AUXM} is 19. Due to its final position, the \textit{RM} is recorded as 0. The above recording process is also applicable to $\langle$\{$c$\}$\rangle$ in $S_4$ and $S_6$. Assuming that $\delta$ is 0.01, the method of calculating the three dimensions is as follows. The current timestamp of $S_3$ \textit{CT}($\langle$\{$c$\}$\rangle$, $S_3$) is 96. The recent timestamp of the first instance acquired from the MT-chain \textit{RT}($\langle$\{($c$, 21)\}:62$\rangle$, $S_3$) is 62, and the recency can be calculated as (1 - 0.01)$^{96-62}$ $\approx$ 0.71. The recent timestamp of the second instance acquired from the MT-chain \textit{RT}($\langle$\{($c$, 19)\}:43$\rangle$, $S_3$) is 43, and the recency can be calculated as (1 - 0.01)$^{96-43}$ $\approx$ 0.59. Thus, the recency of $\langle$\{$c$\}$\rangle$ in $S_3$ is 0.71. As for $\langle$\{$c$\}$\rangle$ in $S_4$ and $S_6$, the recency is both 1, which can be acquired by similar calculations. Hence $R$($\langle$\{$c$\}$\rangle$) = 2.71. The monetary value of $\langle$\{$c$\}$\rangle$ is $M$($\langle$\{$c$\}$\rangle$), and $M$($\langle$\{$c$\}$\rangle$) = $M$($\langle$\{$c$\}$\rangle$, $S_3$) + $M$($\langle$\{$c$\}$\rangle$, $S_4$) + $M$($\langle$\{$c$\}$\rangle$, $S_6$) = \textit{max}\{21, 19\} + \textit{max}\{15, 12\} + 15 = 51. Moreover, there are three elements in the MT-chain, thus, $F$($\langle$\{$c$\}$\rangle$) is 3.

\begin{figure}[ht]
    \centering
    \includegraphics[clip,scale=0.11]{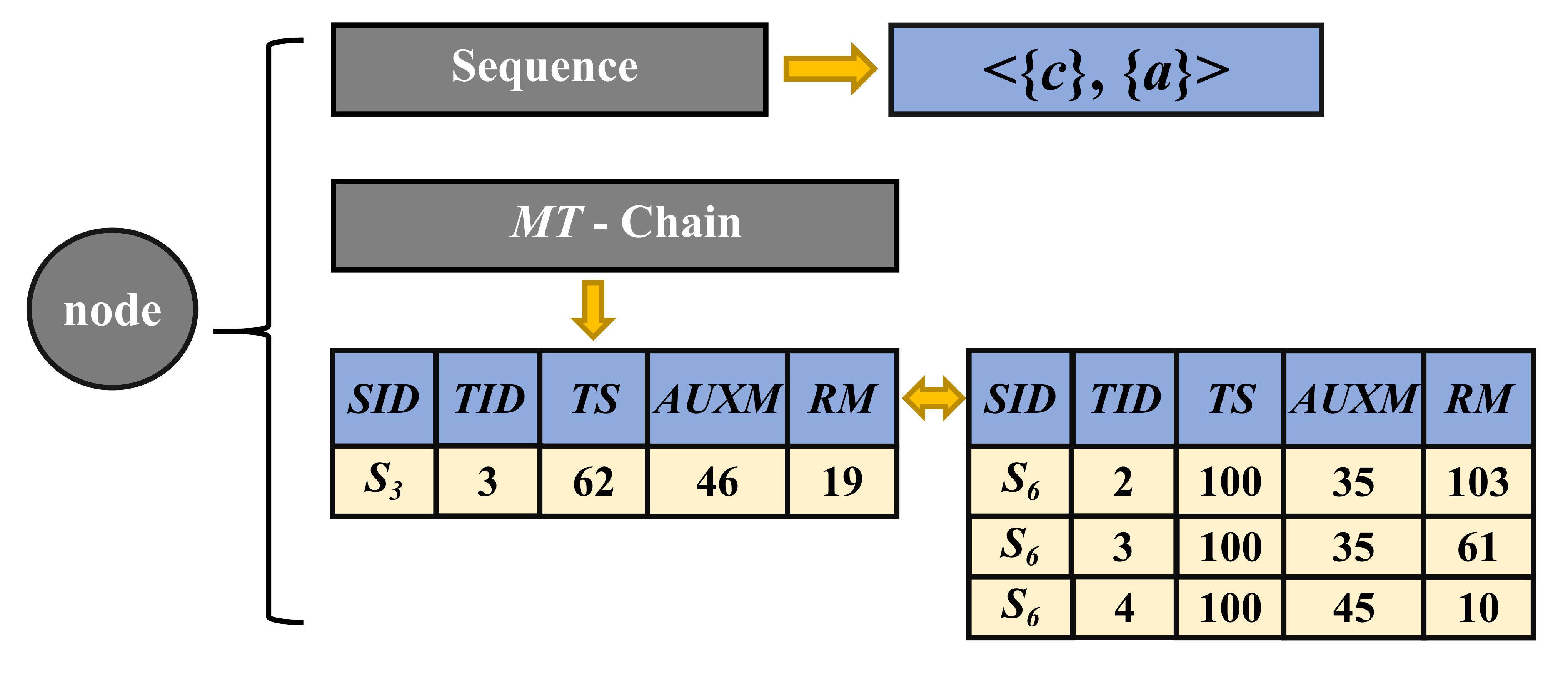}
    \caption{MT-chain for sequence $\langle$\{$c$\}, \{$a$\}$\rangle$.}
    \label{fig: MT-chain of ca}
\end{figure}

To further describe the process of generating an extension pattern through the MT-chain, based on the above example, the process of developing pattern $\langle$\{$c$\}, \{$a$\}$\rangle$ is depicted in Fig. \ref{fig: MT-chain of ca}. In sequence $S_3$, an instance of $\langle$\{$a$\}$\rangle$ occurs in the third itemset as an S-extension, so \textit{TID} is 3. \textit{AUXM}($\langle$\{$c$\}, \{$a$\}$\rangle$, $S_3$) is 21 + 25 = 46. The \textit{RM}($\langle$\{$c$\}, \{$a$\}$\rangle$, $S_3$) is 19. In $S_6$, there are three S-extension positions of 2, 3, and 4, and their corresponding \textit{TID} are 2, 3, and 4. The \textit{AUXM} of instances are 15 + 20 = 35, 15 + 20 = 35, and 15 + 30 = 45, respectively. Moreover, \textit{RM}($\langle$\{$c$\}, \{$a$\}$\rangle$, $S_6$) are 22 + 20 + 21 + 30 + 10 = 103, 21 + 30 + 10 = 61, and 10. Assuming that $\delta$ is 0.01, through a similar process of calculation, the recency value, monetary value, and frequency value can be calculated as 1.71, 91, and 2.

\subsection{Pruning strategies}  \label{sec:strategies}

\begin{definition}[Sequence-weighted monetary \cite{wang2016efficiently}]
    \rm   Given a 1-sequence $Q'$, let $\mathcal{D}$ be a set of sequences $\langle$\{$Q_1$, $Q_2$, $\dots$, $Q_{\zeta}$\}$\rangle$, where $\zeta$ is the size of $\mathcal{D}$. The sequence-weighted monetary value (\textit{SWM}) of $Q'$ is equal to $\sum_{Q_i\sqsubseteq \mathcal{D} \wedge Q'\sqsubseteq Q_i}$ \textit{M}($Q_i$). 
\end{definition}

For example, assuming that $\lambda$ is set to 0.25 and \textit{MMinsup} is 143.75. According to the database in Table \ref{table: database1}, $\langle$\{$d$\}$\rangle$ occurs in $S_1$, $S_2$, and $S_6$, \textit{SWM}($\langle$\{$d$\}$\rangle$). Its sequence-weighted monetary value is \textit{M}($S_2$) + \textit{M}($S_1$) + \textit{M}($S_6$) = 106 + 119 + 138 = 363. Since \textit{SWM}($\langle$\{$d$\}$\rangle$) $\le$ 143.75, we can conclude that $\langle$\{$d$\}$\rangle$ is a \textit{M}-pattern.
    
\begin{theorem}
\label{theorem: SWM}
    \rm Let $Q'$ be a sub-sequence of a sequence $Q$, i.e. $Q' \sqsubseteq$ $Q$ $\wedge$ $Q$ $\sqsubseteq$ $\mathcal{D}$. Then, the monetary value of $Q$ is always less than or equal to the sequence-weighted monetary value of $Q'$, that is \textit{M}($Q'$) $\leq$ \textit{SWM}($Q'$). This property holds for any extension of $Q'$ as well.
\end{theorem}
\begin{proof}
   \rm   Let $C$ be the set of sequences containing instance $Q'$. For any sequence $Q'$ $\sqsubseteq$ $Q$ $\wedge$ $Q$ $\sqsubseteq$ $\mathcal{D}$, it is obvious that \textit{M}($Q'$, $Q$) $\leq$ \textit{M}($Q$).
   Therefore, $\sum_{i=1}^{n}$ \textit{M}($Q'$, $Q_i$) $\le$ $\sum_{i=1}^{n}$ \textit{M}($Q_i$, $C$), where $n$ is the size of $C$. For any extension $Q''$ of $Q'$ in $Q_i$, the property \textit{M}($Q''$, $Q_i$ ) $\le$ \textit{M}($Q_i$, $C$) holds. Thus, \textit{M}($Q'$) $\leq$ \textit{SWM}($Q'$) and \textit{M}($Q''$) $\leq$ \textit{SWM}($Q'$) hold true.
\end{proof}

\begin{strategy}[\textit{RF-SWM} strategy]
\label{strgy: rf-swm}
    \rm According to Theorem \ref{theorem: SWM}, \textit{SWM} is an upper bound on the monetary value of 1-sequences. Combined with the downward closure property in the $R$ dimension \cite{hu2013knowledge} and $F$ dimension \cite{srikant1996mining}, if a sequence satisfies the condition \textit{R}($Q'$) $\ge$ \textit{RMinup} $\wedge$ \textit{F}($Q'$) $\ge$ \textit{FMinup} $\wedge$ \textit{SWM} $\ge$ \textit{MMinup}, then the sequence $Q'$ can be further extended, or the node $t$ of $Q'$ can be pruned safely when one of the above condition is not met.
\end{strategy}

\begin{definition}[The RFM-pattern extension monetary value \cite{wang2016efficiently}]
    \rm  Given a sequence $Q$ containing sub-sequences of $\kappa$ and a set $\mathcal{D}$ whose each sequence contains instances of $\epsilon$, and assuming that the size of $\mathcal{D}$ is $n$, the extension monetary value (\textit{EM}) of a prefix RFM-pattern $\kappa$ at position $p$ is defined as: 
        \begin{equation}
        \textit{EM}(\kappa, Q) = 
            \begin{cases}
                 \textit{TM}(\kappa, p, Q), \; \; {\rm if} \; RM(Q /_{\kappa, p}, Q) \geq 0, \\
                 0,\qquad \qquad \qquad {\rm otherwise}.
            \end{cases}
        \end{equation}
    where $\textit{TM}$ is calculated by $\textit{TM}$($\kappa$, $p$, $Q$) =  \textit{max}\{\textit{AUXM}($\kappa$, $p$, $Q$) + \textit{RM}($\kappa$, $p$, $Q$) $\vert$ $\forall$ $p$ $\wedge$ $\kappa$ $\sqsubseteq$ $Q$\}. We further define that:
    \begin{equation}
        \textit{EM}(\kappa) = \sum_{\forall Q \sqsubseteq \mathcal{D} \wedge \kappa \sqsubseteq Q} \textit{EM}(\kappa, Q).
    \end{equation}
\end{definition}

\begin{theorem}
\label{theorem: EM}
    Given an instance $\kappa$, for each instance $\kappa$ where $\kappa$ is a prefix RFM-pattern of $\epsilon$, \textit{M}($\kappa$) $\le$ \textit{EM}($\kappa$) holds true.
\end{theorem}

\begin{proof}
     \rm   Supposed that instance $\kappa$ can generate $\epsilon$ through I-extension or S-extension, that is, $\kappa$ is the parent node of $\epsilon$, and for each $\kappa$ in the sequence $Q$, \textit{EM}($\kappa$, $Q$) is the maximum monetary value of the sum of \textit{AUXM} and \textit{RM} of each element in the MT-chain. Since sequence $Q$ contains $\kappa$, if \textit{EM}($\kappa$, $Q$) $\leq$ \textit{M}($\kappa$, $Q$), there exists \textit{M}($\kappa$, $Q$) + \textit{RM}($\kappa$, $p_i$, $Q$) $\geq$ \textit{EM}($\kappa$, $Q$), which is inconstant with its definition. Therefore, it always holds that \textit{M}($\kappa$, $Q$) $\le$ \textit{EM}($\kappa$, $Q$). Furthermore, \textit{M}($\kappa$) $\le$ \textit{EM}($\kappa$) holds.
\end{proof}

\begin{strategy}[\textit{RF-EM} strategy]
\label{strgy: rf-em}
   \rm  According to Theorem \ref{theorem: EM}, \textit{EM} is an upper bound of the prefix RFM-pattern. Combined with the downward closure property in the $R$ dimension \cite{hu2013knowledge} and $F$ dimension \cite{srikant1996mining}, the candidate can be safely pruned when there exists \textit{R}($Q'$) $<$ \textit{RMinup} $\vee$ \textit{F}($Q'$) $<$ \textit{FMinup} $\vee$ \textit{EM} $<$ \textit{MMinup} during candidate generation.   
\end{strategy}

\begin{definition}[The prefix RFM-pattern monetary value \cite{wang2016efficiently}]
      \rm  For a sequence $\epsilon$ extended from $\kappa$ in sequence $Q$ at position $p$, we define the monetary value of the prefix pattern (\textit{PM}) of the sequence $\epsilon$ as: 
    \begin{equation}
        \textit{PM}(\epsilon, Q) = 
        \begin{cases}
             \textit{EM}(\kappa, Q), \; {\rm if} \; \kappa \sqsubseteq Q \wedge \epsilon \sqsubseteq Q, \\
             0, \qquad \qquad \qquad {\rm otherwise}.
        \end{cases}
    \end{equation}
    the total residual monetary value of the remaining sequence can be calculated by the formula:
    \begin{equation}
        \textit{PM}(\epsilon) =  \sum_{\forall Q \sqsubseteq \mathcal{D}} \textit{PM}(\epsilon, Q).
    \end{equation}
\end{definition}

For instance, in Fig. \ref{fig: MT-chain of c}, \textit{TM}($\langle$\{$c$\}$\rangle$, $S_3$) is \textit{max}\{21 + 44, 19 + 0\} = \textit{max}\{65, 19\} = 65, \textit{TM}($\langle$\{$c$\}$\rangle$, $S_4$) is max\{15 + 52, 12 + 15\} = \textit{max}\{67, 27\} = 67, and \textit{TM}($\langle$\{$c$\}$\rangle$, $S_6$) is 15 + 123 = 138, respectively. Therefore, \textit{EM}($\langle$\{$c$\}$\rangle$) is equal to 65 + 67 + 138 = 270. In Fig. \ref{fig: MT-chain of ca}, with respect to the extended item $a$, according to the definition, \textit{PM}($\langle$\{$c$\}, \{$a$\}$\rangle$, $S_3$) = \textit{EM}($\langle$\{$c$\}$\rangle$, $S_3$) = 65, \textit{PM}($\langle$\{$c$\}, \{$a$\}$\rangle$, $S_6$) = \textit{EM}($\langle$\{$c$\}$\rangle$, $S_6$) = 138. Therefore, \textit{PM}($\langle$\{$c$\}, \{$a$\}$\rangle$) is the sum of 65 and 138, which is 203.

\begin{theorem}
\label{theorem: PM}
    Given a sequence $\kappa$, which can generate $\epsilon$ through I-extension or S-extension, \textit{M}($\epsilon$) $\le$ \textit{PM}($\epsilon$) holds true.
\end{theorem}
\begin{proof}
     \rm   Based on the definition, for each sequence $\epsilon$ generated by $\kappa$ in sequence $Q$, \textit{PM}($\epsilon$, $Q$) is equal to \textit{EM}($\kappa$, $Q$). In accordance with the theorem \ref{theorem: EM}, the conclusion \textit{M}($\epsilon$) $\le$ \textit{PM}($\epsilon$) holds true if and only if \textit{M}($\kappa$) $\le$ \textit{EM}($\kappa$).
\end{proof}

\begin{strategy}[\textit{PM} strategy]
\label{strgy: pm}
     \rm   According to Theorem \ref{theorem: PM}, the monetary value of the prefix pattern is an upper bound of extension candidates. Therefore, the item can be safely pruned if it satisfies \textit{PM} $\le$ \textit{MMinup} during the process of determining whether to extend or not.
\end{strategy}

\subsection{SeqRFM: Sequential RFM-pattern mining algorithm}

\begin{algorithm}[ht]
    \small
    \caption{The SeqRFM algorithm}
    \label{alg: SeqRFM}
    \LinesNumbered
    \KwIn{The MT-database $\mathcal{D}$; the decay rate $\delta$; the absolute minimum recency threshold $\alpha$; the relative minimum frequency threshold $\beta$; the relative minimum monetary threshold $\gamma$; the maximum time-span length $\theta$.} 
    \KwOut{The set of all RFM-patterns \textit{RFMs}.}
 
    initialize \textit{RMinsup} $\leftarrow$ $\alpha$, \textit{FMinsup} $\leftarrow$ $\beta$ $\times$ $F$($\mathcal{D}$), and \textit{MMinsup} $\leftarrow$ $\gamma$ $\times$ $M$($\mathcal{D}$)\;
    initialize $I$ $\leftarrow$ \{$s$ $\vert$ $s$ $\in$ $\mathcal{D}$ $\wedge$ $s$ {\rm is a 1-sequence $\wedge$ \textit{TSP}($s$) $\le$ $\theta$}\}\;
    \For{{\rm each 1-sequence} $s$ $\in$ $I$}{
        \If{\textit{R}($s$) $<$ \textit{RMinsup} $\vee$ \textit{F}($s$) $<$ \textit{FMinsup} $\vee$ \textit{SWM}($s$) $<$ \textit{MMinsup}}{
        remove $s$ from $I$ and $\mathcal{D}$\;
        }   
        update $I$ and $\mathcal{D}$\;
    } 
    scan $\mathcal{D}$ and construct the node $t$ of each 1-sequence\; 
    add nodes into the set $N^\star$\;
    \For{{\rm each node} $t$ $\in$ $N^\star$}{
        calculate the values of \textit{R(t.seq)}, \textit{F(t.seq)} and \textit{M(t.seq)}\;
        \If{\textit{R(t.seq)} $\ge$ \textit{RMinup} $\wedge$ \textit{F(t.seq)} $\ge$ \textit{FMinup} $\wedge$ \textit{M(t.seq)} $\ge$ \textit{MMinup}}{
            add \textit{t.seq} into \textit{RFMs}\;
        }
        \If{\textit{R(t.seq)} $\ge$ \textit{RMinup} $\wedge$ \textit{F(t.seq)} $\ge$ \textit{FMinup}}{
            \textbf{call} AlgoRecursion($t$, \textit{t.mtchain})\; 
        }
    }
    \textbf{return} \textit{RFMs}
\end{algorithm}

\begin{algorithm}[ht]
    \small
    \caption{AlgoRecursion}
    \label{alg: AlgoRecursion}
    \LinesNumbered
    \KwIn{the node $t$ of the prefix RFM-pattern, the MT-Chain of node $t$.} 

    scan \textit{MT-Chain}: \\
    insert the I-extension items into \textit{ilist} and the S-extension items into \textit{slist}\;
    remove items of low \textit{PM} from \textit{ilist} and \textit{slist}\;
    remove items of large time-span length \textit{TSP} from \textit{ilist} and \textit{slist}\;
    \For{{\rm each node} $ti$ $\in$ \textit{ilist}}{
        \textit{Iseq} = \textit{t.seq} $\bigoplus$ $ts$\;
        calculate the values of \textit{R(\textit{Iseq})}, \textit{F(\textit{Iseq})}, \textit{M(\textit{Iseq})}, \textit{EM}(\textit{Iseq})\;
        \If{\textit{R(\textit{Iseq})} $\geq$ \textit{RMinsup} $\wedge$ \textit{F(\textit{Iseq})} $\geq$ \textit{FMinsup} $\wedge$ \textit{M(\textit{Iseq})} $\geq$ \textit{MMinsup}}{
            add \textit{Iseq} into \textit{RFMs}\;
        }
        \If{\textit{R(\textit{Iseq})} $\geq$ \textit{RMinsup} $\wedge$ \textit{F(\textit{Iseq})} $\geq$ \textit{FMinsup} $\wedge$ \textit{EM(\textit{Iseq})} $\geq$ \textit{MMinsup}}{
            \textbf{call} AlgoRecursion($ti$, \textit{ti.mtchain})\;
        }
    }
    \For{{\rm each node} $ts$ $\in$ \textit{slist}}{
        \textit{Sseq} = \textit{t.seq} $\bigotimes$ $ts$\;
        calculate the values of \textit{R(\textit{Sseq})}, \textit{F(\textit{Sseq})}, \textit{M(\textit{Sseq})}, \textit{EM}(\textit{Sseq})\;
        \If{\textit{R(\textit{Sseq})} $\geq$ \textit{RMinsup} $\wedge$  \textit{F(\textit{Sseq})} $\geq$ \textit{FMinsup} $\wedge$ \textit{M(\textit{Sseq})} $\geq$ \textit{MMinsup}}{
            add \textit{Sseq} into \textit{RFMs}\;
        }
        \If{\textit{R(\textit{Sseq})} $\geq$ \textit{RMinsup} $\wedge$ \textit{F(\textit{Sseq})} $\geq$ \textit{FMinsup} $\wedge$ \textit{EM(\textit{Sseq})} $\geq$ \textit{MMinsup}}{
            \textbf{call} AlgoRecursion($ts$, \textit{ts.mtchain})\;
        }
    }
\end{algorithm}

In Algorithm \ref{alg: SeqRFM}, the process of mining \textit{RFMs} through depth-first search is presented, and the output is the set of all \textit{RFMs}, which are stored in a variable named \textit{RFMs}. First, the algorithm scans the database $\mathcal{D}$ and removes the unpromising 1-sequences based on the \textit{RF-SWM} strategy. Second, the algorithm scans the updated database, constructs nodes, and inserts them into the RFM-tree. Third, SeqRFM begins to scan the nodes of the RFM-tree. By checking whether the pattern of a node $t$ satisfies the three dimensions' thresholds, the algorithm determines whether  \textit{t.seq} should be added to \textit{RFMs} or not. Subsequently, according to the downward closure property in dimensions of $R$ and $F$, the algorithm decides whether to prune or extend the pattern of  \textit{t.seq}. Finally, the algorithm outputs the set of \textit{RFMs} by recursively calling Algorithm \ref{alg: AlgoRecursion}.

As shown in Algorithm \ref{alg: AlgoRecursion}, the process of recursively mining \textit{RFMs} by I-extension and S-extension is performed, where the input is the node $t$ and its MT-chain. First, the MT-chain of node $t$ is scanned, and the I-extension and S-extension items are added to lists, respectively. Then, items with a low \textit{PM} value or a large time-span length \textit{TSP} are removed from lists, which are based on the \textit{PM} strategy or the compact rule. As for I-extensions, the algorithm checks if each extended pattern $\langle$\textit{t.seq} $\bigoplus$ \textit{ti}$\rangle$ satisfies the three dimensions thresholds, to determine if $\langle$\textit{t.seq} $\bigoplus$ \textit{ti}$\rangle$ should be added to \textit{RFMs}. Subsequently, according to the \textit{RF-EM} strategy, the RFM-pattern $\langle$\textit{t.seq} $\bigoplus$ \textit{ti}$\rangle$ is either used in a recursive call of AlgoRecursion to find extended patterns or it is pruned. The algorithm repeats the above process for each item in the list of S-extensions.


\subsection{MSeqRFM: Maximal sequential RFM-pattern mining algorithm}  \label{alg: MSeqRFM}

When dealing with a large database, the patterns with high recency, high frequency, and high monetary values may still be too numerous. Among them are many redundant sequences, which are the sub-sequences of other patterns in the result set. The maximal checking strategy is proposed to filter out the redundant \textit{RFMs} to compress the mining results.

\begin{definition}[Maximal RFM-pattern]
  \rm  A RFM-pattern $P$ is a maximal RFM-pattern (\textit{MRFM}) if none of its super-sequences are \textit{RFMs}. Otherwise, $P$ is not a maximal RFM-pattern.
\end{definition}

For example, given a set $C$ of \textit{RFMs}, assuming that $C$ = $\langle$\{$(abc)$ $(def)$ $(e)$\}, \{$(c)$ $(e)$\}, \{$(de)$\}, \{$(e)$ $(f)$\}$\rangle$, \{$(c)$ $(e)$\} and \{$(de)$\} are sub-sequences of \{$(abc)$ $(def)$ $(e)$\}, so those two sequences are not \textit{MRFMs}, and $(abc)$ $(def)$ $(e)$\} is a maximal pattern in $C$.
    
\begin{strategy}[Maximal checking strategy \cite{fournier2014vmsp}]
     \rm   Given an empty set \textit{MRFMs} which is the set of maximal patterns, the following steps are performed to check whether a new pattern $P$ is maximal or not:
    \begin{enumerate}
        \item Super-sequence checking: if there exists a pattern $P'$ in \textit{MRFMs} such that $P$ is a sub-sequence of $P'$, then $P$ is not a maximal pattern and is excluded from \textit{MRFMs}. Otherwise, $P$ is included in \textit{MRFMs}.
        
        \item Sub-sequence checking: if a pattern $P$ is a super-sequence of any pattern $P'$ in \textit{MRFMs}, then $P'$ is not a maximal pattern and should be removed.
    \end{enumerate}
\end{strategy}

The MSeqRFM algorithm is developed by incorporating the maximal checking strategy in the SeqRFM algorithm, which achieves the goal of mining maximal \textit{RFMs}. Therefore, MSeqRFM is similar to SeqRFM. MSeqRFM uses a boolean function when an RFM-pattern $P$ is added into \textit{RFMs} (lines 13-15 in Algorithm \ref{alg: SeqRFM}, lines 7-9 and 16-18 in Algorithm \ref{alg: AlgoRecursion}), which maintains the set \textit{MRFMs} of maximal \textit{RFMs}. If the boolean function returns true, $P$ is considered as a \textit{MRFM} and added to \textit{MRFMs} afterward. Furthermore, the previous patterns in \textit{MRFMs} are also checked to determine if they are still \textit{MRFMs}. Otherwise, the invalid patterns are removed. The boolean function is shown in Algorithm \ref{alg: MaximalJudge} and the steps are as follows:

\begin{enumerate}[\textbf{Step} 1.]
    \item Differently from  Algorithm \ref{alg: SeqRFM}, a global empty set \textit{MRFMs} is initialized to store all maximal \textit{RFMs}.
    
    \item When a new RFM-pattern $P$ is added to the set of \textit{RFMs}, the pattern $P$ is checked by the boolean function to determine if it is maximal.
\end{enumerate}

\begin{algorithm}[ht]
    \small
    \caption{MaximalJudge}
    \label{alg: MaximalJudge}
    \LinesNumbered 
    \KwIn{the checking pattern $P$.}
    \KwOut{\textbf{true} or \textbf{false}.}
    
    initialize boolean type of \textit{Judge} $\leftarrow$ \textbf{true}\;
    \eIf{\textit{MRFMs} is $\varnothing$}{
        \textbf{return} \textit{Judge}\;
    }{
        \For{{\rm each pattern} $P'$$\in$ \textit{MRFMs}}{
            \If{$P$ $\sqsubseteq$ $P'$}{
                \textit{Judge} $\leftarrow$ \textbf{false}\;
            }
            \If{$P'$ $\sqsubseteq$ $P$}{
                remove $P'$ from \textit{MRFMs}\;
            }
        }
        \textbf{return} \textit{Judge}\;
    }
\end{algorithm}

\section{Experiments} \label{sec: experiments}

To evaluate the effectiveness and memory usage of the proposed SeqRFM algorithm, as well as the data compression capability of the proposed MSeqRFM algorithm, comprehensive experiments are conducted as follows: Firstly, SeqRFM is compared with RFM-PostfixSpan. The evaluation involves a systematic exploration of various parameters, notably the three thresholds, and the specified time-span. These parameter variations are useful in revealing the nuanced performance differences between the two algorithms. Secondly, MSeqRFM is compared to SeqRFM in terms of the number of candidates generated. Note that if an algorithm's memory consumption exceeds the Java heap limit during an experiment, it is regarded as unable to complete the pattern discovery, and its runtime and memory usage are not reported.  

The experimentation was carried out on a 64-bit Windows 11 PC featuring a 12$\rm ^{th}$ Gen Intel (R) Core™ i7-12700F CPU and 16 GB of RAM. All algorithms under consideration were implemented in Java. The source code and datasets are available at GitHub https://github.com/DSI-Lab1/SeqRFM. Further elaboration on these experiments is provided subsequently. 

\subsection{Dataset description and parameters setting}

\textbf{Dataset description.} To evaluate algorithm performance, we employ a combination of three real-life datasets and three synthetic datasets. The BMS, e\_shop, and online\_retail datasets from SPMF\footnote{https://www.philippe-fournier-viger.com/spmf/}  are derived from real-life e-commerce transactions   \cite{fournier2016spmf}. The synthetic datasets are generated by the IBM data generator \cite{agrawal1995mining} with randomly assigned timestamps for the itemsets.

The characteristics of these datasets are shown in Table \ref{table: datasets}, where the symbol $|\mathcal{D}|$ denotes the number of sequences within the dataset; $|I|$ represents the count of distinct items. \textit{Max}($S$) and \textit{Avg}($S$) signify the maximum and average sequence lengths, respectively. The notation \textit{Avg}(\textit{IS}) indicates the average count of itemsets within sequences, and \#\textit{Ele} represents the average item count within each itemset. The following is a brief description of each dataset:

\begin{itemize}
    \item \textbf{BMS}: This dataset contains clickstream data from an e-commerce site. It is a sparse dataset with short sequences and is a single-item-based dataset (each itemset has only a single item).

    \item \textbf{e\_shop}: This dataset contains clickstream data from an online store offering clothing for pregnant women. It is a dense dataset with especially long sequences.

    \item \textbf{online\_retail}: This dataset contains transactions from a UK-based and registered, online store. It is a moderately dense dataset with moderately long sequences.

    \item \textbf{Syn\_20K}: This synthetic dataset has 20,000 sequences. It is a moderately dense dataset with long sequences.

    \item \textbf{Syn\_40K}: This synthetic dataset has twice as many sequences than Syn\_20K, and the number of distinct items 
 is slightly more. It is a moderately dense dataset with moderately long sequences.

    \item \textbf{Syn\_80K}: This is a synthetic dataset with around 80,000 sequences. It is a moderately dense dataset with moderately long sequences.  
\end{itemize}

\begin{table}[ht]
    \renewcommand{\arraystretch}{1.15}
    \caption{The characteristics of the datasets.}
    \label{table: datasets}
    \centering
    \small
    \scalebox{0.9}{ 
    \begin{tabular}{|c|c|c|c|c|c|c|c|}
        \hline
         \textbf{Dataset} & $|\mathcal{D}|$ & $|I|$ & \textit{Max(S)} & \textit{Avg(S)} & \textit{Avg(IS)} & \textit{\#Ele} \\ 
         \hline\hline
        \textbf{BMS} & 59,601 & 497 & 268 & 3.51 & 2.51 & 1.40 \\
         \hline
         \textbf{e\_shop} & 24,026 & 317 & 1755 & 61.99 & 6.89 & 9.00\\ 
         \hline         
         \textbf{online\_retail} & 4,383 & 10,157 & 2403 & 48.23 & 5.36 & 9.00\\ 
         \hline
         \textbf{Syn\_20k} & 20,000 & 7,442 & 213 & 26.98  & 6.22 & 4.34 \\ 
         \hline
         \textbf{Syn\_40k} & 40,000 & 7,537 & 213 & 26.85 & 6.20 & 4.33 \\ 
         \hline
         \textbf{Syn\_80k} & 79,718 & 7,584 & 213 & 26.80 & 6.20 & 4.32  \\ 
         \hline
    \end{tabular}}
\end{table}

\textbf{Parameter setting.} Based on the problem statement, experiments about efficiency analysis and memory analysis focus on the relative frequency threshold $\beta$, and the parameter analysis focuses on the relative monetary threshold $\gamma$ and the time-span length $\theta$. In addition, the value of $\delta$ is set to 0.009 in all experiments. In the following experiments, RFM-PostfixSpan consumes too much memory to discover patterns for small thresholds or in large databases. Moreover, the parameters are set to guarantee that most experimental results can be obtained through multiple adjustments. The parameter setting values of efficiency analysis and memory analysis are listed in Table \ref{table: parameters}.

\subsection{Efficiency analysis}

To evaluate the efficiency of the proposed SeqRFM algorithm,  we use the running time as a performance index to measure the efficiency and compare it with the RFM-PostfixSpan algorithm. Since the frequency dimension is important in sequential pattern mining, and the utility of the HUS-Span algorithm has been evaluated in a prior study \cite{wang2016efficiently}, we conducted experiments with different relative frequency threshold $\beta$ for each dataset.

\begin{table}[ht]
    \caption{The parameter values for efficiency and memory analysis.}
    \label{table: parameters}
    \centering
    \small
    \begin{tabular}{|c|ccccc|}
        \hline
         \textbf{Dataset} & $\delta$ & $\alpha$ & $\beta$ & $\gamma$ & $\theta$ \\ 
         \hline
        \textbf{BMS} & 0.009 & 30 & variable & 0.00002 & 4000 \\
         \textbf{e\_shop} &0.009 & 8000 & variable & 0.04 & 4000\\ 
         \textbf{online\_retail} &0.009 & 300 & variable & 0.02 & 4000\\ 
         \textbf{Syn\_20k} &0.009 & 20 & variable & 0.00002 & 4000 \\ 
         \textbf{Syn\_40k} &0.009 & 70 & variable & 0.00002 & 4000 \\ 
         \textbf{Syn\_80k} &0.009 & 120 & variable & 0.00002 & 4000  \\ 
         \hline
    \end{tabular}
\end{table}

As shown in Fig. \ref{fig: runtime}, the experimental results show that both algorithms have a common trend that as the relative threshold $\beta$ increases, the running time of both algorithms decreases. Moreover, SeqRFM achieves a significant improvement in efficiency over  RFM-PostfixSpan on different datasets. In particular, in datasets BMS, Syn\_20K, and Syn\_40K, SeqRFM takes 66.67\% less time. As the relative threshold $\beta$ decreases, the runtime efficiency of SeqRFM decreases by almost an order of magnitude. In Syn\_80K, SeqRFM performs even better than RFM-PostfixSpan, reducing the time by an order of magnitude in each experiment. That is because SeqRFM has better pruning strategies with tighter upper bounds than RFM-PostfixSpan. The dataset Syn\_80K has far more sequences, so SeqRFM benefits more from using pruning strategies compared to other datasets. For the datasets e\_shop and online\_retail, experimental results cannot be acquired with RFM-PostfixSpan, which also indicates that SeqRFM is capable of discovering valuable patterns when resources are limited. The differences between the two groups of datasets are that e\_shop and online\_retail are dense datasets with moderately long or especially long sequences, which leads to more memory usage for RFM-PostfixSpan and brings out different efficiency improvements. The experimental results indicate that the SeqRFM algorithm effectively improves time efficiency on datasets of different types, especially on datasets with dense and long sequences.

\begin{figure*}[ht]
    \centering
    \includegraphics[clip,scale=0.28]{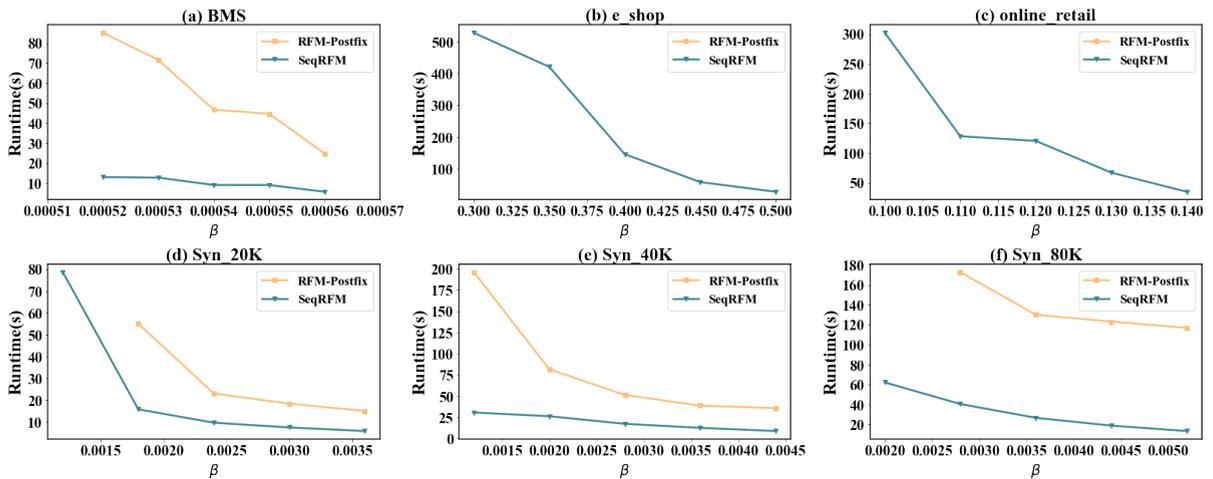}
    \caption{Runtime in each dataset under different values of $\beta$.}
    \label{fig: runtime}
\end{figure*}

\subsection{Memory analysis}

The experimental results of memory usage are shown in Fig. \ref{fig: memory}. As we can see from the figure, as the relative frequency threshold $\beta$ increases, the memory usage decreases in both algorithms. The reason is that as the threshold increases, the number of promising patterns decreases, so the memory required for storing them decreases. In addition, the memory usage of SeqRFM is mostly lower than 10\% of the memory usage of RFM-PostfixSpan among the valid experimental results, which demonstrates the advantage of SeqRFM in memory usage performance. This is because the effective pruning strategies filter out more unpromising extended patterns, which reduces the memory requirement for storing auxiliary information. Note that on datasets e\_shop, online\_retail, Syn\_20K, and Syn\_80K, the blank bar of RFM-PostfixSpan means that the experiments cannot be completed with the available  16 GB of RAM, which illustrates the superiority of SeqRFM to run with limited memory.
The difference in memory usage results is because more information must be stored. Datasets e\_shop and Syn\_80K have a huge number of sequences; accordingly, the algorithms generate more patterns and store more information. The dataset online\_retail is a dense dataset with moderately long sequences that possess multiple patterns found by depth-first search. For the  Syn\_20K dataset, more candidates are generated for low threshold values, which results in finding more patterns and storing more information. The experimental results suggest that the SeqRFM algorithm has lower memory usage than the RFM-PostfixSpan algorithm on datasets of different types, especially large or dense datasets.

\begin{figure*}[ht]
    \centering
    \includegraphics[clip,scale=0.28]{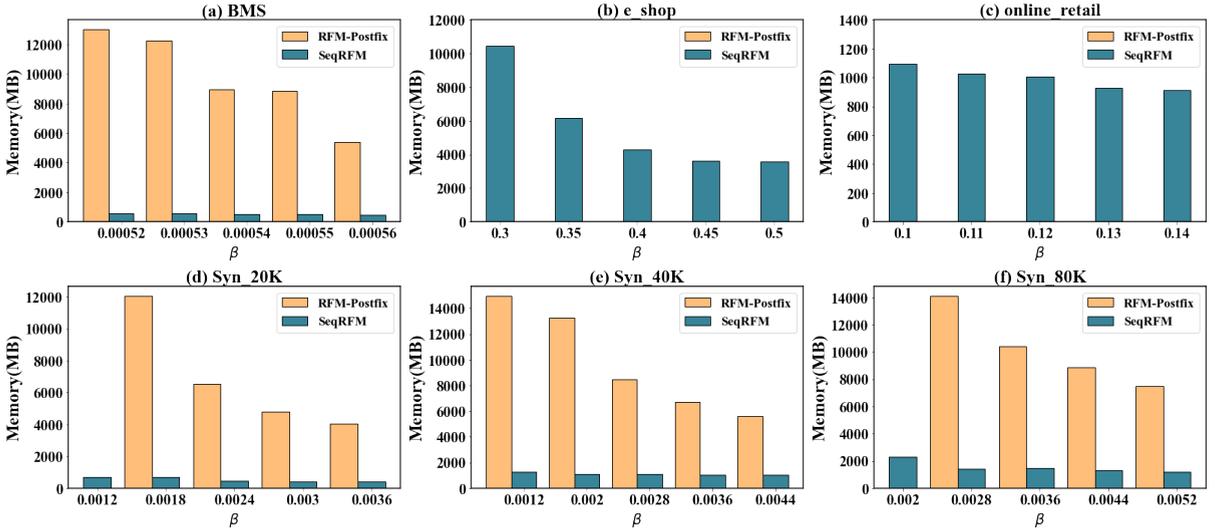}
    \caption{Maximum memory usage in each dataset under different values of $\beta$.}
    \label{fig: memory}
\end{figure*}

\subsection{Parameter analysis}

According to the above experiments based on the variable $\beta$, the SeqRFM algorithm is found to be faster and have a lower memory footprint. In this subsection, we perform additional experiments on distinct datasets to study the effect on SeqRFM of changing the values of the parameters $\gamma$ and $\theta$. The setting of other parameters and the experimental results are shown in Table \ref{table: parameters2}, where \#cand is the number of candidates and \#RFMs means the number of discovered \textit{RFMs}. In these experiments, runtimes are recorded in seconds. If the maximum time-span length $\theta$ is kept constant for all datasets e\_shop, Syn\_20K, Syn\_40K, and Syn\_80K, as the relative monetary threshold $\gamma$ slowly increases, the runtime decreases with some fluctuations. On the BMS and online\_retail datasets, as the relative monetary threshold $\gamma$ rapidly increases, the runtime sharply decreases. If the relative monetary threshold $\gamma$ is kept constant on the six datasets, as the maximum time-span length $\theta$ increases, the runtime decreases considerably. The reasons for the different runtimes of the parameters $\gamma$ and $\theta$ are different. The former affects runtime efficiency by increasing the threshold to prune the candidates. The latter affects runtime efficiency by using the constraint of time-span to reduce the search space. To sum up the above analysis, the relative monetary threshold $\gamma$ and the maximum time-span length $\theta$ are important parameters for the SeqRFM algorithm, which impact the runtime efficiency.

\begin{table*}[ht]
    \renewcommand{\arraystretch}{1.2}
    \caption{The results of parameter analysis.}
    \label{table: parameters2}
    \centering
    \small
    \resizebox{\linewidth}{!}{ 
    \begin{tabular}{|c|c|c|c|c|c|c|c|c|c|c|c|c|c|}
        \hline
         \textbf{Dataset} & \multicolumn{13}{c|}{\textbf{BMS}} \\ 
         \hline
         \multirow{4}{*}{\shortstack{\\ $\delta$: 0.009 \\ $\alpha$: 20 \\$\beta$: 0.000375}} & 
         $\theta$ & 225 & 237 & 250 & 225 & 237 & 250 & 225 & 237 & 250 & 225 & 237 & 250   \\ \cline{2-14} &
         $\gamma$ & \multicolumn{3}{c|}{0.015} & \multicolumn{3}{c|}{0.018} & \multicolumn{3}{c|}{0.021} & \multicolumn{3}{c|}{0.024} \\ \cline{2-14} &
         Runtime & 503.39 & 956.97 & 1893.87 & 470.59 & 874.89 & 1795.53 & 319.39 & 423.79 & 823.35 & 50.56 & 64.89 & 87.77  \\ \cline{2-14} &
         \#cand & 173,776,663 & 340,336,562 & 670,662,489 & 130,628,043 & 251,177,789 & 501,365,834 & 68,202,052 & 87,744,776 & 179,992,903 & 8,317,600 & 10,102,151 & 13,788,361   \\ \cline{2-14} &
         \#RFMs & 10 & 10 & 10 & 5 & 5 & 5 & 4 & 4 & 4 & 2 & 2 & 2    \\ 
         \hline
         \textbf{Dataset} & \multicolumn{13}{c|}{\textbf{e\_shop}} \\ 
         \hline
         \multirow{4}{*}{\shortstack{\\ $\delta$: 0.009 \\ $\alpha$: 7000 \\$\beta$: 0.01}} & 
         $\theta$ & 25 & 50 & 75 & 25 & 50 & 75 & 25 & 50 & 75 & 25 & 50 & 75   \\ \cline{2-14} &
         $\gamma$ & \multicolumn{3}{c|}{0.002} & \multicolumn{3}{c|}{0.01} & \multicolumn{3}{c|}{0.018} & \multicolumn{3}{c|}{0.026} \\ \cline{2-14} &
         Runtime & 36.03 & 114.85 & 258.13 & 36.50 & 113.90 & 256.48 & 35.62 & 114.38 & 257.75 & 35.42 & 114.45 & 258.74   \\ \cline{2-14} &
         \#cand & 11,793 & 38,774 & 80,545 & 11,792 & 38,773 & 80,545 & 11,787 & 38,758 & 80,535 & 11,779 & 38,695 & 80,453   \\ \cline{2-14} &
         \#RFMs & 270 & 857 & 1,756 & 262 & 849 & 1,748 & 165 & 742 & 1,641 & 76 & 486 & 1,366 \\ 
         \hline
         \textbf{Dataset} & \multicolumn{13}{c|}{\textbf{online\_retail}} \\ 
         \hline
         \multirow{4}{*}{\shortstack{\\ $\delta$: 0.009 \\ $\alpha$: 220 \\$\beta$: 0.05}} & 
         $\theta$ & 100 & 125 & 150 & 100 & 125 & 150 & 100 & 125 & 150 & 100 & 125 & 150 \\ \cline{2-14} &
         $\gamma$ & \multicolumn{3}{c|}{0.01} & \multicolumn{3}{c|}{0.02} & \multicolumn{3}{c|}{0.03} & \multicolumn{3}{c|}{0.04} \\ \cline{2-14} &
         Runtime & 358.62 & 691.75 & 1,502.09 & 333.49 & 646.40 & 1341.55 & 293.38 & 547.60 & 1106.47 & 244.31 & 463.85 & 940.12   \\ \cline{2-14} &
         \#cand & 9,245,222 & 16,346,343 & 32,835,195 & 5,798,574 & 11,195,592 & 19,845,914 & 3,144,273 & 5,890,170 & 9,629,340 & 1,572,272 & 2,773,010  & 4,613,414   \\ \cline{2-14} & 
         \#RFMs & 71,762 & 134,559 & 276,432 & 3,629 & 7,606 & 11,543 & 100 & 157 & 211 & 5 & 9 & 9    \\ 
         \hline
         \textbf{Dataset} & \multicolumn{13}{c|}{\textbf{Syn\_20k}}\\ 
         \hline
         \multirow{4}{*}{\shortstack{\\ $\delta$: 0.009 \\ $\alpha$: 8 \\$\beta$: 0.001}} & 
         $\theta$ & 25 & 50 & 75 & 25 & 50 & 75 & 25 & 50 & 75 & 25 & 50 & 75  \\ \cline{2-14} &
         $\gamma$ & \multicolumn{3}{c|}{0.00002} & \multicolumn{3}{c|}{0.0001} & \multicolumn{3}{c|}{0.00018} & \multicolumn{3}{c|}{0.00026} \\ \cline{2-14} &
         Runtime & 52.70 & 172.00 & 212.25 & 46.07 & 140.70 & 167.04 & 44.90 & 136.24 & 151.89 & 44.11 & 138.57 & 152.71  \\ \cline{2-14} &
         \#cand & 94,832,533 & 337,551,411 & 407,677,058 & 16,857,270 & 46,660,605 & 53,615,616 & 6,249,141 & 21,880,667 & 25,294,075 & 4,785,440 & 18,064,774 & 20,583,735   \\ \cline{2-14} &
         \#RFMs & 505,202 & 2,682,236 & 3,907,037 & 465,352 & 2,619,226 & 3,841,330 & 356,694 & 2,324,105 & 3,520,585 & 228,418 & 1,764,072 & 2,843,951   \\ 
         \hline
         \textbf{Dataset} & \multicolumn{13}{c|}{\textbf{Syn\_40k}}\\ 
         \hline
         \multirow{4}{*}{\shortstack{\\ $\delta$: 0.009 \\ $\alpha$: 40 \\$\beta$: 0.001}} & 
         $\theta$ & 25 & 50 & 75 & 25 & 50 & 75 & 25 & 50 & 75 & 25 & 50 & 75 \\ \cline{2-14} &
         $\gamma$ & \multicolumn{3}{c|}{0.00002} & \multicolumn{3}{c|}{0.0001} & \multicolumn{3}{c|}{0.00018} & \multicolumn{3}{c|}{0.00026} \\ \cline{2-14} &
         Runtime & 84.22 & 220.42 & 267.58 & 68.03 & 182.58 & 215.26 & 67.21 & 183.73 & 202.99 & 71.99 & 176.79 & 201.76 \\ \cline{2-14} &
         \#cand & 79,610,307 & 234,731,470 & 286,142,608 & 5,119,973 & 12,988,886 & 15497995 & 2,913,746 & 9,412,869 & 11,061,562 & 2,299,488 & 7,532,794 & 8,824,773  \\ \cline{2-14} &
         \#RFMs & 205,690 & 951,546 & 1,358,901 & 185,668 & 918,125 & 1,324,131 & 142,103 & 793,848 & 1,181,293 & 93,628 & 597,567 & 921,319   \\ 
         \hline
         \textbf{Dataset} & \multicolumn{13}{c|}{\textbf{Syn\_80k}}\\ 
         \hline
         \multirow{4}{*}{\shortstack{\\ $\delta$: 0.009 \\ $\alpha$: 40 \\$\beta$: 0.001}} & 
         $\theta$ & 25 & 50 & 75 & 25 & 50 & 75 & 25 & 50 & 75 & 25 & 50 & 75 \\ \cline{2-14} &
         $\gamma$ & \multicolumn{3}{c|}{0.00002} & \multicolumn{3}{c|}{0.0001} & \multicolumn{3}{c|}{0.00018} & \multicolumn{3}{c|}{0.00026} \\ \cline{2-14} &
         Runtime & 207.96 & 470.49 & 554.58 & 160.72 & 424.40 & 501.75 & 162.13 & 422.61 & 480.11 & 156.01 & 419.81 & 432.80  \\ \cline{2-14} &
         \#cand & 66,136,885 & 145,742,979 & 172,303,173 & 3,951,868 & 10,599,962 & 12,476,117 & 2,933,490 & 8,100,840 & 9,610,559 & 2,466,995 & 7,081,263 & 8,506,619  \\ \cline{2-14} &
         \#RFMs & 242,672 & 1,029,769 & 1,360,957 & 213,976 & 986,495 & 1,316,334 & 155,915 & 845,592 & 1,162,660 & 94,274 & 614,354 & 895,668 \\ 
         \hline
    \end{tabular}}
\end{table*}

\subsection{Compression ability of maximal algorithm}

To evaluate the compression ability of the MSeqRFM algorithm, several experiments are conducted on the six datasets of Table \ref{table: datasets}. We compare the number of patterns mined by the two algorithms. Moreover,  to guarantee a large number of mining patterns for the practical situation in big data mining, the total number of mining results of \textit{RFMs} in each experiment must be four figures or more. The compression effect is shown in Table \ref{table: MRFMs}. In the experiment, MSeqRFM can greatly compress the \textit{RFMs} for the highest compression rate. For instance, in Syn\_20K, Syn\_40K, and Syn\_80K, the compression rate is up to 98\% or more, and that in e\_shop is also up to around 87\%. Especially in BMS, the compression rate is low, which is 48.39\%. The reason for the different results is that BMS is a single-item-based and sparse dataset consisting of short sequences. In each sequence, each itemset with a single item is ordered by item name, and each item is unique in each sequence. Therefore, if the single-item RFM-pattern does not have a super-sequence RFM-pattern in other sequences, it is added to the set of results. That indicates a low redundancy rate in the mining of BMS, which leads to a low compression rate.

In summary, MSeqRFM achieves the goal of result compression using the maximal checking strategy. In particular, MSeqRFM has good compression ability in the multiple-item-based dataset and can reduce redundancy to a great extent. Even in the single-item-based dataset, MSeqRFM can compress around half of the redundant patterns.

\begin{table}[ht]
    \caption{The number of maximal patterns and all patterns.}
    \label{table: MRFMs}
    \centering
    \small
    \begin{tabular}{|c|ccc|}
        \hline
         \textbf{Dataset} & \#RFMS & \#MRFMs & Compression rate \\ 
         \hline
        \textbf{BMS} & 94,054 & 48,542 & 48.39\%  \\ 
         \textbf{e\_shop} & 4,770  & 627 & 86.86\%  \\ 
         \textbf{online\_retail} & 23,785 & 1,040 & 95.63\% \\ 
         \textbf{Syn\_20k} & 246,546 & 2,051 & 99.17\% \\ 
         \textbf{Syn\_40k} & 193,342 & 2,146 & 98.89\% \\ 
         \textbf{Syn\_80k} & 140,155 & 2,418 & 98.27\%  \\ 
         \hline
    \end{tabular}
\end{table}

\section{Conclusion and Future Work}  \label{sec: conclusion}

In this paper, we propose an advanced method named SeqRFM to increase the efficiency of RFM-pattern mining in sequential data. The fundamental concepts and the computation methods of the three dimensions are redefined, which makes the mining process more rigorous. In addition, three upper bounds on the monetary dimension and several pruning strategies are proposed to reduce unnecessary calculations. Furthermore, we develop an efficient algorithm named MSeqRFM to compress the mining results. Experiments on multiple aspects have demonstrated the superiority of SeqRFM in time and memory, the difference in parameter comparison in different dimensions, and the compression ability of MSeqRFM. There are several potential avenues for future research and extensions. A possibility is to further extend the method by considering additional pruning strategies to improve the efficiency and execution time of the mining process. Furthermore, for a more comprehensive exploration, other dimensions besides recency, frequency, monetary value, and time-span length can be considered in the future. Additionally, further research can focus on larger datasets to evaluate the algorithm's performance and assess its efficiency and applicability in handling big data scenarios.

\bibliographystyle{cas-model2-names}
\bibliography{SeqRFM.bib}

\end{document}